\documentclass[12pt]{amsart}
\usepackage{url}
\usepackage{amsmath,amsxtra,amssymb,latexsym,epsfig,amscd,amsthm,fancybox,epsfig}
\usepackage{amsfonts}
\usepackage{mathtools}
\usepackage{float}
\restylefloat{figure}
\usepackage{booktabs}
\usepackage{multirow}
\usepackage{color}
\usepackage{placeins}
\usepackage{lscape}
\usepackage{mathtools}
\usepackage{amsmath}
\usepackage[subfigure]{graphfig}

\usepackage[mathscr]{eucal}
\usepackage{graphicx}
\usepackage{tikz}
\usetikzlibrary{calc}
\usepackage{caption}
\usepackage{multicol,xcolor}
\usepackage{epsfig} % for postscript graphics files
\usepackage{epstopdf}
\usepackage{cases}
\usepackage{color}
\usepackage{hyperref}
\usepackage{bm}  %JLK
\usepackage{bbm} %JLK
\usepackage[ruled,vlined]{algorithm2e}

\newtheorem{thm}{Theorem}[section]
\newtheorem{prop}{Proposition}[section]
\newtheorem{lem}{Lemma}[section]

\newtheorem{remark}{Remark}[section]

\newtheorem {asp}{Assumption}[section]

\theoremstyle{definition}

\theoremstyle{remark}

\numberwithin{equation}{section}

\numberwithin{equation}{section}

\newcommand{\bed}{\begin{displaymath}}
\newcommand{\eed}{\end{displaymath}}
\newcommand{\bea}{\bed\begin{array}{rl}}
\newcommand{\eea}{\end{array}\eed}

\newcommand{\barray}{\begin{array}{ll}}
\newcommand{\earray}{\end{array}}
\newcommand{\diag}{{\rm diag}}

\newcommand{\1}{\boldsymbol{1}}

\DeclareMathOperator*{\argmax}{arg\,max}

\def\a.s{\text{\;a.s.\;}}

\newcommand{\beq}[1]{\begin{equation} \label{#1}}
\newcommand{\eeq}{\end{equation}}

\title[]{Maximum Likelihood Estimation of Diffusions by Continuous Time Markov Chain}
\author[J.L. Kirkby]{J. Lars Kirkby }
\address{School of Industrial and Systems Engineering \\
Georgia Institute of Technology\\
Atlanta, GA 30318, USA}
\email{jkirkby3@gatech.edu}

\author[D.H. Nguyen]{Dang H. Nguyen }
\address{Department of Mathematics \\
University of Alabama\\
345 Gordon Palmer Hall\\
Box 870350 \\
Tuscaloosa, AL 35487-0350 \\
United States}
\email{dangnh.maths@gmail.com}

\author[D. Nguyen]{Duy Nguyen }
\address{Department of Mathematics\\
Marist College\\
3399 North Road\\
Poughkeepsie NY 12601\\
United States
}
\email{nducduy@gmail.com}

\author[N. Nguyen]{Nhu N. Nguyen }
\thanks{Nhu N. Nguyen was in part supported by
	the National Science Foundation
	under grant DMS-1710827.}
\address{Department of Mathematics \\
University of Connecticut\\
341 Mansfield\\
 Storrs, CT 06269\\
United States}
\email{nguyen.nhu@uconn.edu}
\keywords{MLE, diffusion, SDE, Maximum Likelihood estimation, CTMC, Continuous Time Markov Chain, stochastic differential equation, estimation}
\subjclass[2010]{34D20, 60H10, 92D25, 93D05, 93D20.}

\begin{document}
\maketitle

\begin{abstract}
 In this paper we present a novel method for estimating the parameters of a parametric diffusion processes. Our approach is based on a closed-form Maximum Likelihood estimator for an approximating Continuous Time Markov Chain (CTMC) of the diffusion process. Unlike typical time discretization approaches, such as psuedo-likelihood approximations with Shoji-Ozaki or Kessler's method, the CTMC approximation introduces no time-discretization error during parameter estimation, and is thus well-suited for typical econometric situations with infrequently sampled data. Due to the structure of the CTMC, we are able to obtain closed-form approximations for the sample likelihood which hold for general univariate diffusions. 
  Comparisons of the state-discretization approach with approximate MLE (time-discretization) and Exact MLE (when applicable) demonstrate favorable performance of the CMTC estimator.  Simulated examples are provided in addition to real data experiments with FX rates and constant maturity interest rates.
\end{abstract}
%\newpage
\tableofcontents
\newpage

\section{Introduction}
%\blue{Target Journals? Computational Statistics and data analysis,  J. Statistics and Computing, ??}
Diffusion processes are used extensively in 
financial engineering to model
the dynamics of stock prices, interest rates, and foreign exchange rates, among numerous other applications.
Notable examples include the geometric Brownian motion (GBM) which is used 
in the Black-Scholes framework to model the stock prices \cite{black1973pricing} or the
 Cox–Ingersoll–Ross (CIR) diffusion, which was introduced by \cite{cox2005theory} to describes the evolution of interest rates, and is also widely used as a model for stochastic volatility.
 
 The drift  and  diffusion terms of the process contain
 a number of unknown parameters which require estimation before
 the model can be used. Typically, a discretized finite sample path of the process
 is collected and used to estimate the unknown parameters. For a continuous time diffusion,
 its transition density function plays a crucial role in understanding the dynamics
 of the process. Most importantly, it can be used to estimate the model's unknown parameters by means of maximum likelihood. Unfortunately, the transition
 density function is unknown for most diffusion models, and it becomes virtually impossible to determine the exact maximum likelihood estimates for the unknown parameters.
To overcome the unavailability of the transition density, approximation methods are  usually employed.
 Several econometric approaches
 have been proposed to estimate the unknown parameters. 
 These econometric methods can be categorized as the simulation approach (\cite{gourieroux1993indirect,gallant1996moments}), (generalized) method of moments (\cite{hansen1993back,kessler1999estimating}),
 (non)parametric density matching (\cite{ait1995nonparametric,ait1996testing}, and Bayesian methodologies (\cite{eraker2001mcmc,jones1997bayesian}). \cite{ait2002maximum} makes a fruitful breakthrough
 in using Hermite polynomials to orthogonally approximate
 the transition density of a univariate time-homogeneous diffusion.
 This idea was later extended to time-inhomogenous diffusion in \cite{egorov2003maximum},
 multivariate time-homogenous (\cite{ait2008closed}) and time-inhomogeneous diffusions (\cite{choi2013closed}), stochastic volatility (\cite{ai2007maximum}) and affine multi-factor models (\cite{ait2010estimating}).

Note that to apply the method of  \cite{ait2002maximum,ait2008closed} for univariate diffusions, one must
be able to transform the given diffusion to a unit diffusion process where the volatility is the identity.
The method is inapplicable if the diffusion is not reducible in this fashion, although the reducible condition was later relaxed in the work of \cite{choi2015explicit} and recently of \cite{yang2019new}.
For further extensions, please see the recent work of \cite{li2013maximum}, \cite{li2019pricing}, and references therein.

%\red{TODO(Duy): Reword this}
Diffusion processes evolve continuously both in
space and time. As a result,
their analytical tractability
is usually limited except for some very special cases, and efficient and accurate numerical approximation
is often employed for statistical inference. 
%Since diffusion
%processes evolve continuously in time and in state, and there is usually limited analytical tractability except for a few very special cases, thus an efficient and accurate
%approximation method is needed.
In general, there are two possible directions for
approximating a diffusion process: 1) time discretization which includes the Euler discretization as well
as higher order time-stepping schemes (see \cite{higham2001algorithmic,jacod2011discretization} for a comprehensive account
of existing methods), and 2) spatial discretization, where one can discretize the state space into a finite discrete grid of spatial points, while preserving the continuous time dimension of the diffusion
process. We take the later approach and approximate the evolution of the diffusion process through a
continuous-time Markov chain (CTMC).
Research in this direction was initiated in  \cite{mijatovic2013continuously}
where the authors
approximated
the value of barrier options
under a very general time-homogenous and time-inhomogenous one-dimensional Markov process.
This idea was later extended to price Asian options (\cite{cai2015general,kirkby2020efficient}), and realized variance derivatives (\cite{cui2017general}) under stochastic volatility dynamics.
% There is recent interest in using a continuous-time Markov chain (CTMC) to approximate a one-dimensional Markov process and then valuate options based on this approximation framework. \cite{mijatovic2013continuously} considered the pricing of continuously monitored barrier options and established rigorous convergence properties for the approximation. \cite{cai2015general} priced discretely and continuously monitored arithmetic Asian options through numerical inversion of a double transform. Recently \cite{song2016unified} extended the valuation framework to a one-dimensional Markov process with regime-switching coefficients. Extension to time-inhomogenous case is considered in \cite{mijatovic2013continuously,ding2021markov}.
A rigorous error analysis for the CTMC approximation and the optimal discretization grid design was considered in  \cite{li2018error,zhang2019analysis}.  Simulation of two-dimensional diffusions was proposed in \cite{cui2021efficient}.

In this paper, by approximating a univariate diffusion by a continuous time Markov chain (CTMC), we present a novel method for estimating the parameters of a parametric diffusion process. To the author's best knowledge, this is the first time a CTMC approximation has been used to conduct statistical inference. Our approach is based on a closed-form Maximum Likelihood estimator for the approximating CTMC, and it requires no ``reducibility conditions'' or specific knowledge of the process (other than the parametric form of the model family) to be applicable.  The breadth of processes that are well approximated is large, including those encountered in finance, economics, and the physical sciences.  The present work focuses on the important case of univariate diffusions, with extensions to the multivariate case left for future analysis.  Because the approach is immune to time-discretization bias, it can be safely applied in situations with infrequent time-sampling, such as weekly, monthly, quarterly, or yearly sampled time-series, which are all common in econometric data.  Time-discretization approaches, such as Euler's method or Shoji-Ozaki, should be applied with caution in such cases, as their validity hinges upon a small time-step.

The rest of this paper is organized as follows: Section \ref{sect:Problem formulation}
introduces the problem of estimating parameters
of one-dimensional diffusions.
In Section \ref{sect:Continuous time Markov chain approximation},
we consider approximating
the one-dimensional diffusion
by a finite state CTMC by  explicitly
constructing the transitional matrix which governs
the dynamics of the CMTC.
Section \ref{sect:MLECTMC} is concerned 
with the Maximum Likelihood Estimation (MLE)
of parameters of the diffusion. 
We provide a rigorous convergence analysis as well as 
a quasi-Newton method which is used to numerically approximate
the unknown parameters. In Section \ref{section:Numerical Examples},
we provide numerous
numerical examples
to demonstrate the effectiveness
of the proposed method, including a real data example using 10-Year Constant Maturity interest rates, and another using foreign exchange data.
We also compare the obtained results with various
numerical approaches in the literature. Comparisons with existing approaches, including Exact MLE when applicable, demonstrate that the method is quite reliable.
Section \ref{section:Conclusion} concludes the paper.

%%%%%%%%%%%%%%%%%%%%%%%%%%%%%%%
\section{Problem formulation}\label{sect:Problem formulation}
%%%%%%%%%%%%%%%%%%%%%%%%%%%%%%%
Consider the stochastic diffusion process 
\begin{equation}\label{eq:Diffuse}
dS_t=\mu(S_t,\theta)dt+\sigma(S_t,\theta)dW_t, \quad t\geq 0,
\end{equation}
where $(W_t)_{t\geq 0}$ is the standard Brownian motion, $\mu(S_t,\theta): (\mathbb R \times \mathbb R^d)\rightarrow \mathbb R$ and $\sigma(S_t,\theta): (\mathbb R \times \mathbb R^d)\rightarrow \mathbb R_+$
are the drift and diffusion term, respectively. The unknown parameter vector $\theta=(\theta_1,\theta_2,\ldots,\theta_d)$
belongs to a compact set $\Theta\subset \mathbb R^d$. We will assume that the drift and diffusion functions satisfy the local Lipschitz condition with linear growth, which guarantees a weakly unique solution to \eqref{eq:Diffuse}.
For example, recall the Geometric Brownian Motion (GBM)
\begin{equation}
dS_t=\mu S_t dt+\sigma S_tdW_t,
\end{equation}
where in this case $\theta=(\mu,\sigma)$ is unknown. Given a time-step $\Delta>0$, let $p(\Delta,s^\prime,s)\equiv p(\Delta,s^\prime,s;\theta)$ denote the transition density function of $S_t$; that is
\begin{equation}
\mathbb P(S_{t+\Delta}\in ds^\prime|S_t=s)=p(\Delta, s^\prime,s)ds^\prime.
\end{equation}
To estimate the unknown parameter $\theta$, we assume that
a discrete sample of $S_t$ is observed: $S_1,S_2,\ldots,S_N$, with observations taken at a uniform frequency $\Delta$. By the Markovian property of $S_t$, the sample log-likelihood function is given by
\begin{equation}
L_N(\theta, \Delta):=\sum_{n=1}^{N-1}\ln p(\Delta,S_{n+1},S_{n}).
\end{equation}
The maximum likelihood estimator (MLE) of $\theta$ is defined to be the maximizer of the following constrained optimization problem:
\begin{equation}
\widehat{\theta}_N:=\argmax_{\theta\in\Theta} L_N(\theta, \Delta),
\end{equation}
and we will refer to this as the \emph{Exact MLE}, as it utilizes the exact transition density.
Similar to \cite{li2013maximum},  we make the following assumption throughout this paper.
\begin{asp}
The transition density $p(\Delta,x,y)$ is
continuous in $\theta\in\Theta$ and the log-likelihood function
$ L_N(\theta,\Delta)$ admits a unique
maximizer in the parameter set $\Theta$.
\end{asp}
It is well known that a closed-form expression for
$p(\Delta, s^{\prime},s)$ is unavailable for most diffusion processes. Thus, in most cases it can be difficult
to find $\widehat{\theta}_N$ explicitly, and approximations are typically used \cite{hurn2007seeing}. A notable example is the Hermite expansion approach of \cite{ait2002maximum, ait2008closed}, see also \cite{choi2015explicit}.  Other standard examples include the method of Kessler \cite{kessler1997estimation} as well as that of Shoji-Ozaki \cite{shoji1998statistical}. For an excellent overview of estimation methods for SDE, including the above-mentioned approaches, see \cite{iacus2009simulation}.

In the next section, we propose a
continuous time Markov chain approximation to the diffusion $S_t$,
and using the CTMC process, we can approximate $\widehat{\theta}_N$ in a straightforward manner.  A key advantage of this methodology is that it can be applied with only a knowledge of the parametric form of the diffusion. It does not rely on our ability to derive closed-form expressions for specific models, or any other tricks to make it applicable in practice.  As we discussed below, it even has computational advantages over exact/approximate MLE for large samples.

%\textbf{Problem:} Assume that we can observe a discrete sample $S_1,S_2,\ldots,S_N$, sampled at a uniform frequency $\Delta$.
%Our goal is to use this observed sample
%to estimate $\theta$.
%
%\textbf{Ideas:} We will approximate $S_t$ by a continuous time Markov chain (CTMC)
%called $S_t^{m}$, where $m$ is the number of states.
%We will then carry out the Maximum likelihood estimation on this  CTMC process.

\section{Continuous time Markov chain approximation}
\label{sect:Continuous time Markov chain approximation}
The essence of this work is to define a tractable approximation to the diffusion in \eqref{eq:Diffuse} for which the likelihood is available in closed form. We accomplish this via a general CTMC approximation of the diffusion. 
 
\subsection{The CTMC}\label{sect:Binning}
Given a parametric diffusion family characterized by \eqref{eq:Diffuse}, we will construct
a continuous-time Markov chain $\{S^{m}_t\}_{t\geq 0}$, taking values in some discrete state-space $\mathbb{S}_{m}:=\{s_1,s_2,\ldots,s_{m}\}$, whose dynamics well resemble those of $S_t$. 
For the Markov chain $S_t^{m}$,  its transitional dynamics are described
by the \textit{rate matrix} $\mathbf Q=\mathbf Q(\theta)=[q_{ij}(\theta)]_{m\times m}\in\mathbb R^{m\times m}$, whose elements $q_{ij}=q_{ij}(\theta)$ satisfy the $q$-property: 
(i) $q_{ii}\leq 0$, $q_{ij}\geq 0$ for $i\neq j$,
and (ii) $\sum_{j}q_{ij}=0, \forall i=1,2,\ldots,m$. In terms of $q_{ij}$'s, the transitional probability of the CTMC $S_t^{m}$ is given by: 
\begin{equation}
\label{transitionalProp}
\mathbb{P}(S^{m}_{t+\Delta}=s_j| S^{m}_{t}=s_i, S_{t^{\prime}}^{m},0\leq t^{\prime}\leq t)=
\delta_{ij}+q_{ij}\Delta+o(\Delta^2),
\end{equation}
where in the above expression $\delta_{ij}$ denotes the Kronecker delta.  In particular, the transitional matrix is represented  in the form of a matrix exponential:
\begin{equation}\label{eq:StateTrans}
\mathbf T(\Delta) = \exp(\mathbf Q(\theta) \Delta) = \sum_{k=0}^\infty (\mathbf Q(\theta) \Delta)^k/(k!), \quad \Delta> 0.
\end{equation}
% Here the finite set $\mathbb{S}_{m}$, which is the state space of the CTMC $\{S^{m}_t\}_{t\geq 0}$,  is carefully chosen
%such that the state space of $S_t$ is sufficiently covered. Details on how to choose the grid points $s_1,s_2,\ldots,s_{m}$ are given in Section \ref{sect:Binning}. In addition, 
%the construction must guarantee that $S^{m}_t$ weakly converges to its continuous counterpart $S_t$ under 
%appropriate technical conditions. This is particularly helpful since it guarantees that the desired expected values of well behaved path functionals converge to the true values as the grid points are made denser in the space of $S_t$. 

%%%%%%%%%%%%%%%%%%%%%%%%%
%\subsection{Binning The Sample}\label{sect:Binning}
%%%%%%%%%%%%%%%%%%%%%%%%%
\begin{figure}[h!t!b!]
\centering     %%% not \center
\subfigure{\includegraphics[width=.51\textwidth]{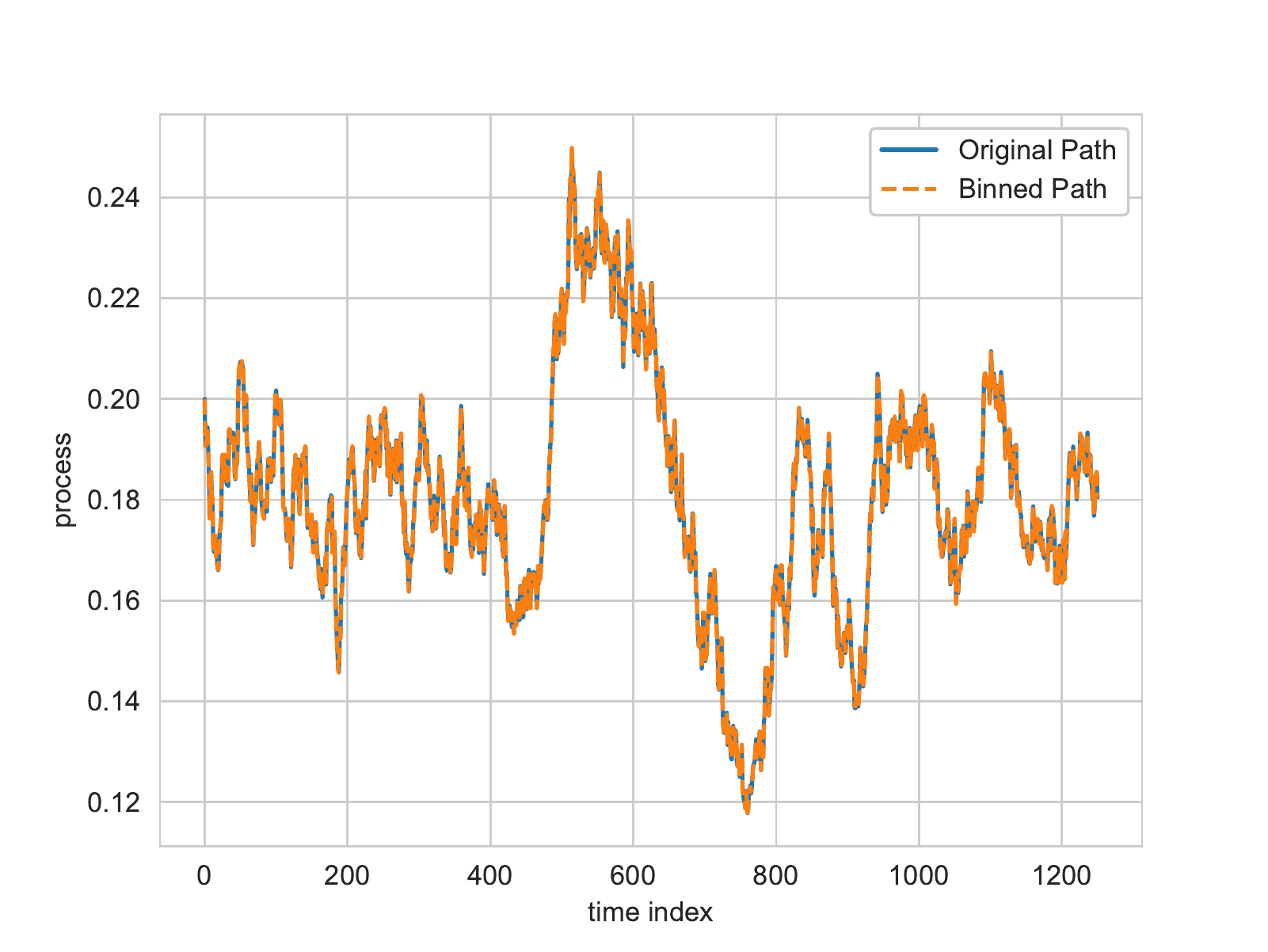}}\hspace{-1.8em}
\subfigure{\includegraphics[width=.51\textwidth]{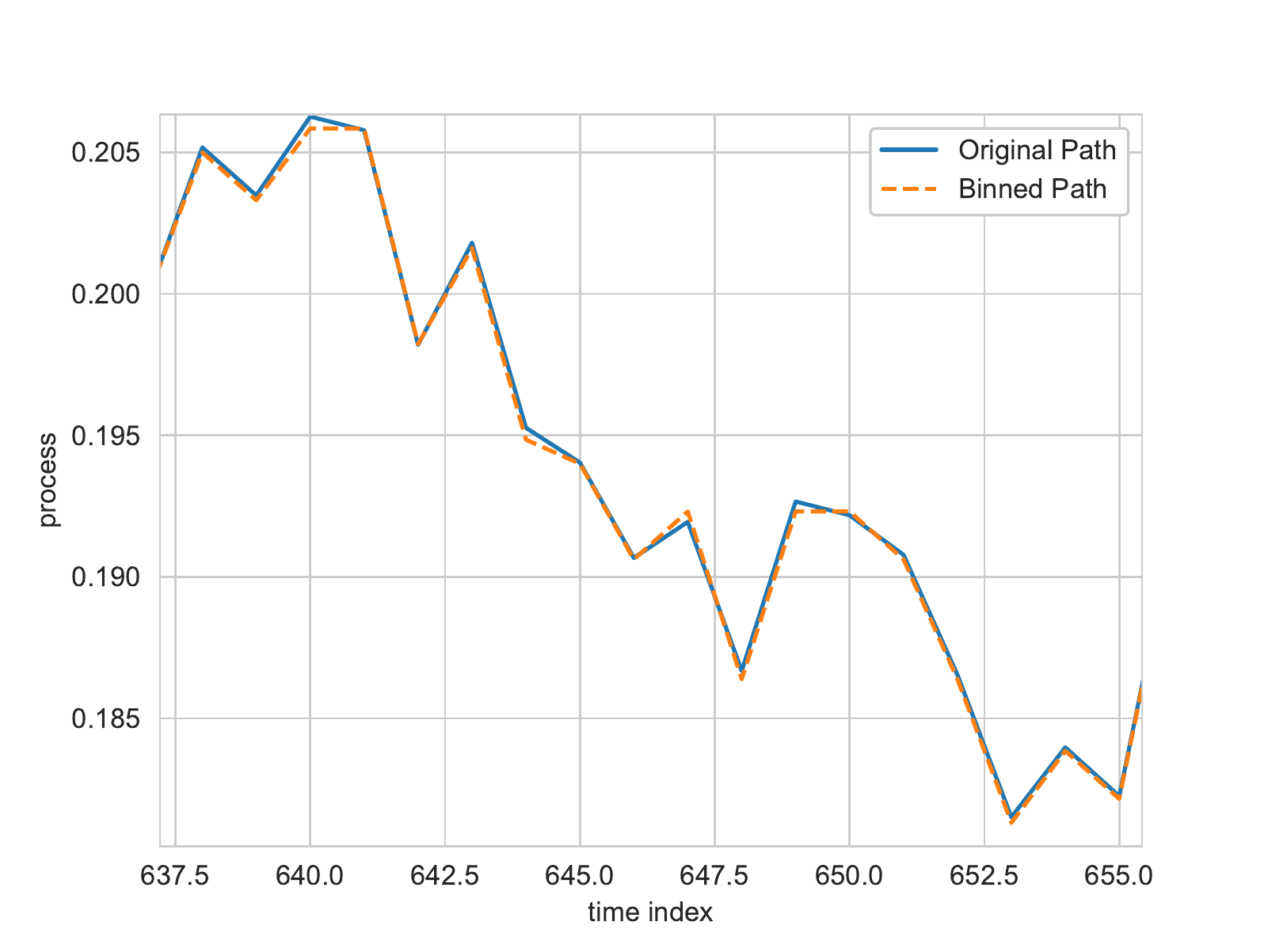}}
\caption{\small{Simulated OU path and binned sample (Left), with zoom-in (Right). Params: $S_0=0.2$, $\kappa = 2$, $m=0.2$, $\sigma=0.15$.
}}\label{fig:OUPath}
\end{figure}
 
Suppose that we are given a continuous sample $\bm S = (S_1, S_2, \ldots, S_N)$. We first construct a state-space $\bm S_m$, based on the observed sample. For simplicity, we consider a uniformly spaced state-space which covers the full sample range.\footnote{We could also apply a non-uniform space, which clusters more points around high density areas of the sample.} We then apply a simple binning procedure to obtain a mapped sample that lives in the state-space of the CTMC. Specifically, we map $s_{j} = \mathcal I(S_{k\Delta})$ where $s_j$ is the nearest point in the CTMC state-space to $S_{k\Delta}$, and we denote this value by $S^m_{k\Delta}$. In this way we observe the \emph{discrete} sample $\bm S^m:=(S^m_1,S^m_2,\ldots
,S^m_N)$, where each $S^m_n\in \mathbb{S}_{m}$.
This procedure is illustrated in Figure \ref{fig:OUPath} for the Ornstein-Uhlenbeck (OU) model, given by $dS_t=\kappa(\mu-S_t)dt+\sigma dW_t$.  Zooming in on the sample path in the right panel, we can see that with sufficiently many states, the continuous sample path is well approximated by the discrete state-space.  
%The transition matrix $\mathbf T(\Delta)$ in \eqref{eq:StateTrans} permits us to estimate 

%\begin{remark}[State-Space]
%TODO: complete. Right now we just use a uniform state-space, extending a bit beyond $[\min\{S_n\}, \max\{S_n\}]$. However, this can be wasteful of states, so I will also try the following. Let $m$ be the number of states, and let $F_{N}(x)$ denote the empirical cdf. Let $0=u_1,u_2, \ldots, u_N=1$ denote a uniform grid on $[0,1]$. The state space is then $s_i:=F_N^{-1}(u_i)$, $i=1,\ldots, N$, linear interpolating over all flat regions. I will extend it a bit beyond the min and max, similar to the initial approach.
%\end{remark}

%%%%%%%%%%%%%%%%%%%%%%%%
\subsection{Approximating the Diffusion Generator}
%%%%%%%%%%%%%%%%%%%%%%%%
In Section \ref{sect:Binning}, we constructed a state-space $\mathbb{S}_{m}$ for the CTMC, and mapped the continuous sample path onto $\mathbb{S}_{m}$. The next step in the process is to define the generator matrix $\bm Q(\theta)$ so that the continuous-time dynamics are well matched by the CTMC. The generator is parameterized by $\theta$ through the parametric family we have chosen as the model. Determination of the actual values of $\theta$, and hence the makeup of $\bm Q(\theta)$, will be accomplished via maximum likelihood estimation,  discussed in Section \ref{sect:MLECTMC}.

We will require a few more concepts.
First, for a bounded Borel function $H$, define
\begin{equation}
\label{P-operator}
P_tH(s):=\mathbb{E}_s[H(S_t)]:=\mathbb{E}[H(S_t)|S_0=s],
\end{equation}
%Throughout this paper, it is assumed
and recall  that $S$ satisfies the Markov property:
\begin{equation}
\label{MarkovProperty}
\mathbb E[H(S_{t+r})|\mathcal{F}_t]=P_rH(S_t).
\end{equation}
From \eqref{MarkovProperty}, the family of operators $(P_t)_{t\geq 0}$ is easily seen to form
a semigroup:
\begin{equation}
\label{semiGroup-Property}
P_{t+r}H=P_t(P_rH),\quad \forall r,t\geq 0,\quad\mbox{and}\quad P_0H=H.
\end{equation}

Let $C_0(\mathcal{S})$ denote the set of continuous functions on the state space $\mathcal{S}$ that vanish at infinity. To guarantee the existence of a \textit{version} of $S$
with c\'{a}dl\'{a}g paths satisfying the (strong) Markov process, we assume the following Feller's properties:
\begin{asp}
$S=\{S_t\}_{t\geq 0}$ is a Feller process on $\mathcal{S}$. That is, for any
$H\in C_0(\mathcal{S})$, the family of operators $(P_t)_{t\geq 0}$ satisfies
\begin{itemize}
\item $P_tH\in C_0(\mathcal{S})$ for any $t\geq 0$;
\item $\lim_{t\to 0}P_tH(s)=H(s)$ for any $s\in\mathcal{S}$.
\end{itemize}
\end{asp}

\noindent The family $(P_t)_{t\geq 0}$ is determined by its infinitesimal generator $\mathcal{L}$, where
\begin{equation}
\label{S-generator-def}
\mathcal{L}H(s):=\displaystyle\lim_{t\to 0^+}\frac{P_tH(s)-H(s)}{t},\quad \forall H\in C_0(\mathcal{S}).
\end{equation}
For the diffusion given in \eqref{eq:Diffuse}, we have 
\begin{equation}
\label{S-generator}
\mathcal{L}H(s)=\frac{1}{2}\sigma^2(s,\theta)\frac{\partial ^2 H}{\partial s^2}+
\mu(s,\theta)\frac{\partial H}{\partial s}. 
\end{equation}

The CTMC approximation is based on approximating the generator $\mathcal{L}$ by $\mathcal{L}^{m}$, defined as follows. For each $i\in \{1,2,\ldots,m-1\}$ define $k_i:=s_{i+1}-s_i$, and let $\mu^{+}$($\mu^{-}$) denote respectively the positive (negative) part of the function $\mu$. A non-uniform finite discretization of  $\mathcal{L}H(x)$ in \eqref{S-generator} is given by:
\begin{align}
\label{discrete-generator}
&\mu(s_i)\left(\displaystyle\frac{-k_i}{k_{i-1}(k_{i-1}+k_i)}H(s_{i-1})
+\displaystyle\frac{k_i-k_{i-1}}{k_ik_{i-1}}H(s_i)
+\displaystyle\frac{k_{i-1}}{k_i(k_{i-1}+k_i)}H(s_{i+1})
 \right)\notag\\
&+\displaystyle\frac{ \sigma^2(s_i)}{2}
\left(\frac{2}{k_{i-1}(k_{i-1}+k_i)}H(s_{i-1})
-\frac{2}{k_{i-1}k_i}H(s_i)
+\frac{2}{k_i(k_{i-1}+k_i)}H(s_{i+1})
 \right)\notag\\
 &=q_{i,i-1}H(s_{i-1})+q_{i,i}H(s_i)+q_{i,i+1}H(s_{i+1})\notag\\
& =:\mathcal{L}^{m}H(s).
\end{align}

\noindent where $q_{i,j}$'s are chosen as in \cite{lo2014improved}, which is  recalled here
\begin{equation}
\label{q-MCA-Heston}
q_{ij}(\theta)=\left\{
\begin{array}{ll}
\displaystyle\frac{\mu^-(s_i,\theta)}{k_{i-1}}+\displaystyle\frac{\sigma^2(s_i,\theta)-(k_{i-1}\mu^-(s_i,\theta)+k_i\mu^+(s_i,\theta))}{k_{i-1}(k_{i-1}+k_i)},& \;\; \text{if}\;\; j=i-1,\\
\displaystyle\frac{\mu^+(s_i,\theta)}{k_{i}}+\displaystyle\frac{\sigma^2(s_i,\theta)-(k_{i-1}\mu^-(s_i,\theta)+k_i\mu^+(s_i,\theta))}{k_{i}(k_{i-1}+k_i)},& \;\; \text{if}\;\; j=i+1,\\
-q_{i,i-1}-q_{i,i+1},&\;\; \text{if}\;\;j = i,\\
0, & \;\; \text{if}\;\; j\neq i-1,i,i+1.\\
\end{array}
\right.
\end{equation}
Here  $\textbf{k}:=\{k_1,k_2,\ldots,k_{m-1}\}$ is assumed to be chosen such that
\[
0<\displaystyle\max_{1\leq i\leq m-1}\{k_i\}\leq\displaystyle \min_{\theta\in \Theta}\min_{1\leq i\leq m}\left\{\frac{\sigma^2(s_i,\theta)}{|\mu(s_i,\theta)|}\right\}.
\]
\begin{remark}
With this choice of $k_i$'s, $\mathbf Q(\theta)=[q_{ij}]_{m\times m}$ is a tridiagonal matrix. Moreover, we  have
\begin{align}\label{local-consistency}%
\sigma^2(s_i)\geq  \displaystyle\max_{1\leq i\leq m-1}\{k_i\}\cdot|\mu(s_i,\theta)|&\geq  \displaystyle\max_{1\leq i\leq n-1}\{k_i\}\cdot(\mu^+(s_i,\theta)+\mu^-(s_i,\theta))\nonumber\\
&\geq k_{i-1}\mu^-(s_i,\theta)+k_i\mu^+(s_i,\theta).
\end{align}
As a result, the $q$-property is satisfied:  $q_{ij}\geq 0,\forall 1\leq i\neq j\leq m$, and $\sum_{j=1}^{m}q_{ij}=0,  i=1,\ldots,m.$  
\end{remark}

\begin{remark} (Boundary Conditions)
For the diffusion $S_t$ with state space $\mathbf{S}=(s_l,s_r)$ ($-\infty\leq s_l<s_r\leq\infty$), we assume that the two endpoints are \textit{inaccessible} if $\mathbf{S}$ is an infinite interval. Otherwise, the boundary points can be classified as \textit{exit, entrance}, or \textit{natural}. {Note that a \textit{natural} endpoint does not belong to the state space since it can not be reached in finite time. See discussions on page 15 of \cite{borodin2012handbook}.} A \textit{non-singular} boundary point is both entrance and exit, and it can be further classified into \textit{reflecting} or \textit{absorbing}. For further details, please refer to page 16-17 of \cite{borodin2012handbook}. 
When constructing the continuous-time Markov chain approximation ${S}_t^{m}$, we assume that the boundary points $s_1,s_{m}$ are reflecting or absorbing.
\end{remark}
Under some appropriate conditions, it can be shown that  $S_t^{m}$ converges weakly to $S_t$ as $m\to\infty$.
More specifically, there is the following result. 
\begin{thm}(Weak convergence \cite{mijatovic2013continuously})
Let $S$ be a Feller process whose infinitesimal generator
$\mathcal{L}$ does not vanish at zero and infinity.
Let $S_t^{m}$ be the continuous time Markov chain with the generator given in \eqref{discrete-generator}.
Assume that $\max_{s\in\mathbb{S}_{m}}|\mathcal LH(s)-\mathcal{L}^{m}H(s)|\to 0$ as $m\to\infty$ for all functions $H$ in the core of $\mathcal L$  and 
$\lim_{s\to 0^+}\mathcal L H(s)=0$, then $S_t^{m}$ converges weakly to $S_t$ as $m\to\infty$.
That is, $\mathbb E[H(S_T^{m})|S_0]\to\mathbb E[H(S_T)|S_0]$ for all bounded continuous functions $H$. 
\end{thm}

%%%%%%%%%%%%%%%%%%%%%%%%
\section{MLE Estimator}\label{sect:MLECTMC}
%%%%%%%%%%%%%%%%%%%%%%%%
In this section we
construct the MLE estimate for $\theta$ based on the observed sample.
Let $\Delta>0$ and assume that we observe $\bm S^m:=(S^m_1,S^m_2,\ldots
,S^m_N)=(S^m_{\Delta},S^m_{2\Delta},\ldots,S^m_{N\Delta})$.  In particular, we assume that we have applied a binning procedure, as outlined in Section \ref{sect:Binning}, to arrive at sample belonging to the state space of the approximating CTMC.

Define the $m\times m$ probability transition matrix
\begin{equation}
\mathbf T(\Delta)=\exp(\mathbf Q\Delta)=\sum_{i=0}^{\infty}\frac{(\mathbf Q\Delta)^i}{i!}.
\end{equation}
Note that since our $\mathbf Q=\mathbf Q(\theta)$ is a function of $\theta$
so  is $\mathbf T(\Delta)$, and $\mathbf T(\Delta)_{ij}$ is the transition probability from the state $s_i$ to state $s_j$.
The likelihood of the sample is given by
\begin{equation}\label{eq:PSms}
P(\bm S^m|S^m_1,\mathbf Q)=\prod_{n=1}^{N-1} \mathbf T(\Delta)_{S^m_{n\Delta},S^m_{(n+1)\Delta}}.
\end{equation}
Here $\mathbf T(\Delta)_{S^m_{i\Delta},S^m_{(i+1)\Delta}}$ corresponds to $\mathbf T(\Delta)_{j,k}$, with $j=\mathcal I(S^m_{i\Delta})$ and $k = \mathcal I(S^m_{(i+1)\Delta})$, where we define the index mapping
\[
\mathcal I: \mathbb S_m \rightarrow \{1,\ldots, m\},
\]
which maps $\mathcal I(S^m_j)\rightarrow j$, the corresponding state index.
%\red{With binning:
%\begin{equation}
%P(S|S_1,\mathbf Q)=\prod_{i=1}^{N-1} \mathbf T(\Delta)_{ \widetilde S_{i\Delta},\widetilde S_{(i+1)\Delta}}.
%\end{equation}
%}
As in \cite{kalbfleisch1985analysis,mcgibbon2015efficient}, let $\mathbf C(\Delta)\in \mathbb N^{m\times m}$ be the matrix such that
\begin{equation}\label{eq:Cmat}
\mathbf C(\Delta)_{i,j}=\sum_{n=1}^{N-1}\1_{\{S^m_{n\Delta}=s_i\}}\cdot\1_{\{S^m_{(n+1)\Delta}=s_j\}},
\end{equation}
which counts the number of times in the sample that a transition from state $s_i$ to $s_j$ occurs.
We can then see from \eqref{eq:PSms} that
$$
P(\bm S^m|\mathbf Q,S^m_1)=\prod_{1\leq i,j \leq m}\mathbf T(\Delta)_{i,j}^{C(\Delta)_{i,j}}.
$$
The log likelihood function is 
\begin{align}\label{eq:Like}
 L_{N,m}(\theta,\Delta)&=\ln P(\bm S^m|\mathbf Q(\theta),S^m_1)\nonumber\\
&=\sum_{i,j}\mathbf C(\Delta)_{i,j}\ln \mathbf T(\Delta)_{i,j}\nonumber\\
&=\sum_{i,j}\left(\mathbf C(\Delta)\circ\ln \exp(\Delta \mathbf Q(\theta) \right)_{i,j}.
\end{align}
Here $\circ$ denotes the Hadamard matrix product
and $\ln(\mathbf A)$ is the element-wise logarithm.
The maximum likelihood estimator (MLE) is
\begin{equation}
\widehat{\theta}_{N,m}=\argmax_{\theta\in\Theta} L_{N,m}(\theta,\Delta).
\end{equation}

Next let's consider the eigendecomposition of $\mathbf Q$:
$$
\mathbf Q=\mathbf V\mathbf\Lambda\mathbf U^T,
$$
where the columns
of $\mathbf U$ and $\mathbf V$ are formed by the left and right
eigenvectors of $\mathbf Q$, respectively. Moreover, $\mathbf U^T=\mathbf V^{-1}$.
$\mathbf\Lambda=\diag(\lambda_1,\lambda_2,\ldots,\lambda_n)$
is a diagonal matrix formed by the set of eigenvalues of $\mathbf Q$.
Theorem \ref{diagonalizable-theorem} plays a crucial role in computing the probability transition matrix $\mathbf T(\Delta)$.
\begin{thm}\label{diagonalizable-theorem}
The tridiagonal matrix $\mathbf Q$ defined in \eqref{q-MCA-Heston} is diagonalizable.  In addition, $\mathbf Q$ has exactly $m$ simple  real eigenvalues satisfying $0\geq \lambda_1>\lambda_2>\ldots>\lambda_{m}$. Hence, the transitional matrix $\mathbf T(\Delta)$ has the following decomposition:
\begin{align}\label{eq:DiagRep}
 \mathbf T(\Delta)&=\bm V e^{\mathbf \Lambda \Delta} \bm U^T\quad\mbox{with}\quad \mathbf Q=\mathbf V\mathbf\Lambda\mathbf U^T,
\end{align}
where
$\mathbf \Lambda=\diag({\lambda_1,\lambda_2,\ldots,\lambda_{m}})$ is a diagonal matrix of the eigenvalues of $\mathbf Q$.
%, $\bm \Gamma=(\gamma_{ij})_{i,j=1,\ldots,m}$ is a matrix whose columns are the corresponding eigenvectors, and we write $\bm \Gamma^{-1}=(\tilde{\gamma}_{ij})_{i,j=1,\ldots,m}$. 
\end{thm}
\begin{proof}
The part that $\mathbf Q$ has exactly $m$ distinct eigenvalues can be found in  \cite{cui2019general}.
For the second claim, let $u=(u_1,u_2,\ldots,u_m)$ be an
eigenvector corresponding to the eigenvalue $\lambda$ of $\mathbf Q$. We will
show that $\lambda\leq 0$. To this end, by choosing $i$ such that $|u_i|=\max\{|u_j|:j=1,\ldots,m\}$,  we have
$$
\displaystyle\sum_{j\neq i}^{m}\lambda_{ij}u_j=\lambda  u_i-\lambda_{ii}u_i.
$$
This implies that
$$
|\lambda  -\lambda_{ii}|\leq \displaystyle\sum_{j\neq i}^{m}\lambda_{ij}|\frac{u_j}{u_i}|\leq \sum_{j\neq i}^{m}\lambda_{ij}.
$$
Hence $\lambda$ is in the circle centered at $\lambda_{ii}$ with radius $\sum_{j\neq i}^{m}\lambda_{ij}$. Since  $\lambda_{ii}\leq 0$ and  $\sum_{j}^{m}\lambda_{ij}=0$, we have $\lambda\leq 0$. This completes the proof.
\end{proof}

\begin{figure}[h!t!b!]
\centering     %%% not \center
\subfigure[Brownian Motion]{\includegraphics[width=.51\textwidth]{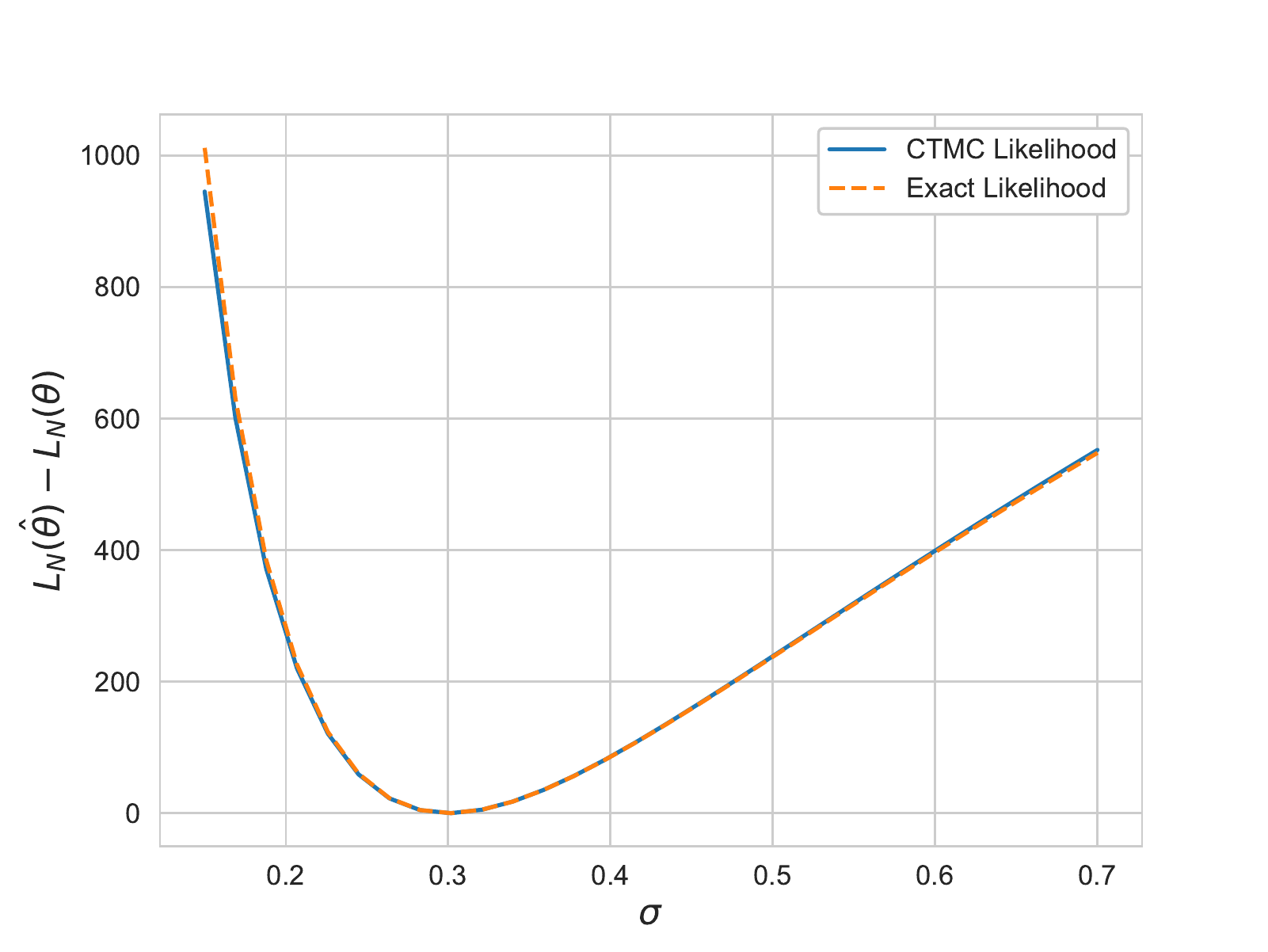}}\hspace{-1.8em}
\subfigure[Ornstein-Uhlenbeck]{\includegraphics[width=.51\textwidth]{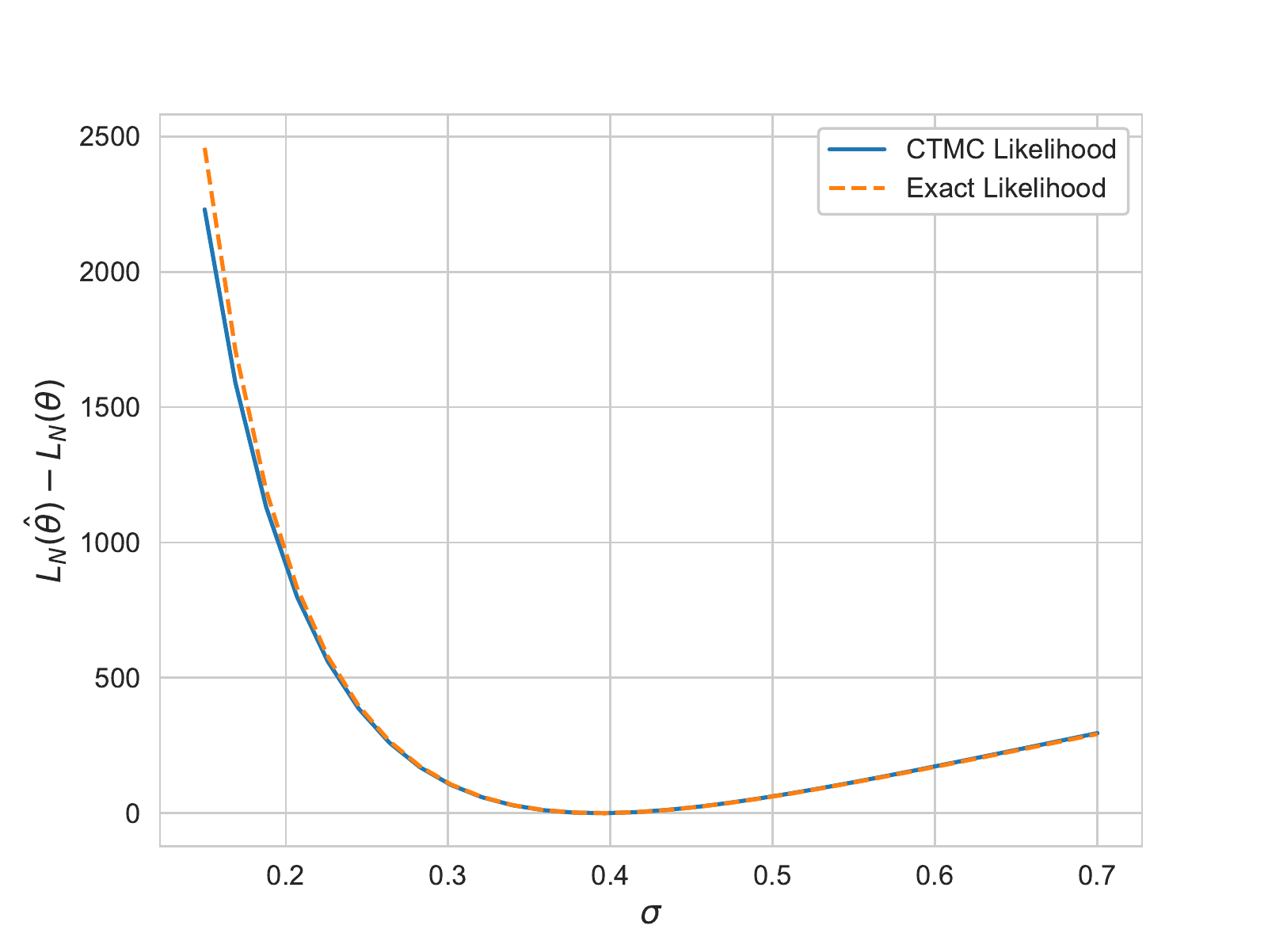}}
\caption{\small{Likelihood comparison of CTMC approximation vs. exact likelihood, as a function of single diffusion parameter $\sigma$. Sample size $N=1250$, $\Delta=1/250$.  (Left) Brownian motion: $S_0=10,\mu=0.08, \sigma=0.3$. (Right). OU: $S_0=0.2$, $\kappa = 4$, $\mu=0.2$, $\sigma=0.4$.
}}\label{fig:LogLike}
\end{figure}

\noindent\textbf{Example:} To make the CTMC approximation concrete, we consider two diffusion examples for which the transition probability density is known in closed-form, which permits the use of \emph{exact} maximum likelihood via $
L_N(\theta, \Delta):=\sum_{n=1}^{N-1}\ln p(\Delta,S_{n+1},S_{n}).
$ First we simulate the drifted Brownian motion, $dS_t = \mu dt + \sigma dW_t$, and second the Ornstein-Uhlenbeck (OU) model with $dS_t=\kappa(\mu-S_t)dt+\sigma dW_t$. In both cases the transition density is Gaussian. Figure \ref{fig:LogLike} displays the exact Likelihood function versus the CTMC approximation \eqref{eq:Like} with $m=250$ states. More specifically, we normalize each using $-L_N(\theta, \Delta) + L_N(\theta^*, \Delta)$, where $\theta^*=\theta^{MLE}$ is the MLE parameter set. It is clear from the figure that the CTMC approximation is very accurate compared with the exact transition density, and results in a tight approximation to the likelihood function, and hence the MLE estimate. 

\begin{figure}[h!t!b!]
\centering     %%% not \center
\subfigure[Brownian Motion]{\includegraphics[width=.531\textwidth]{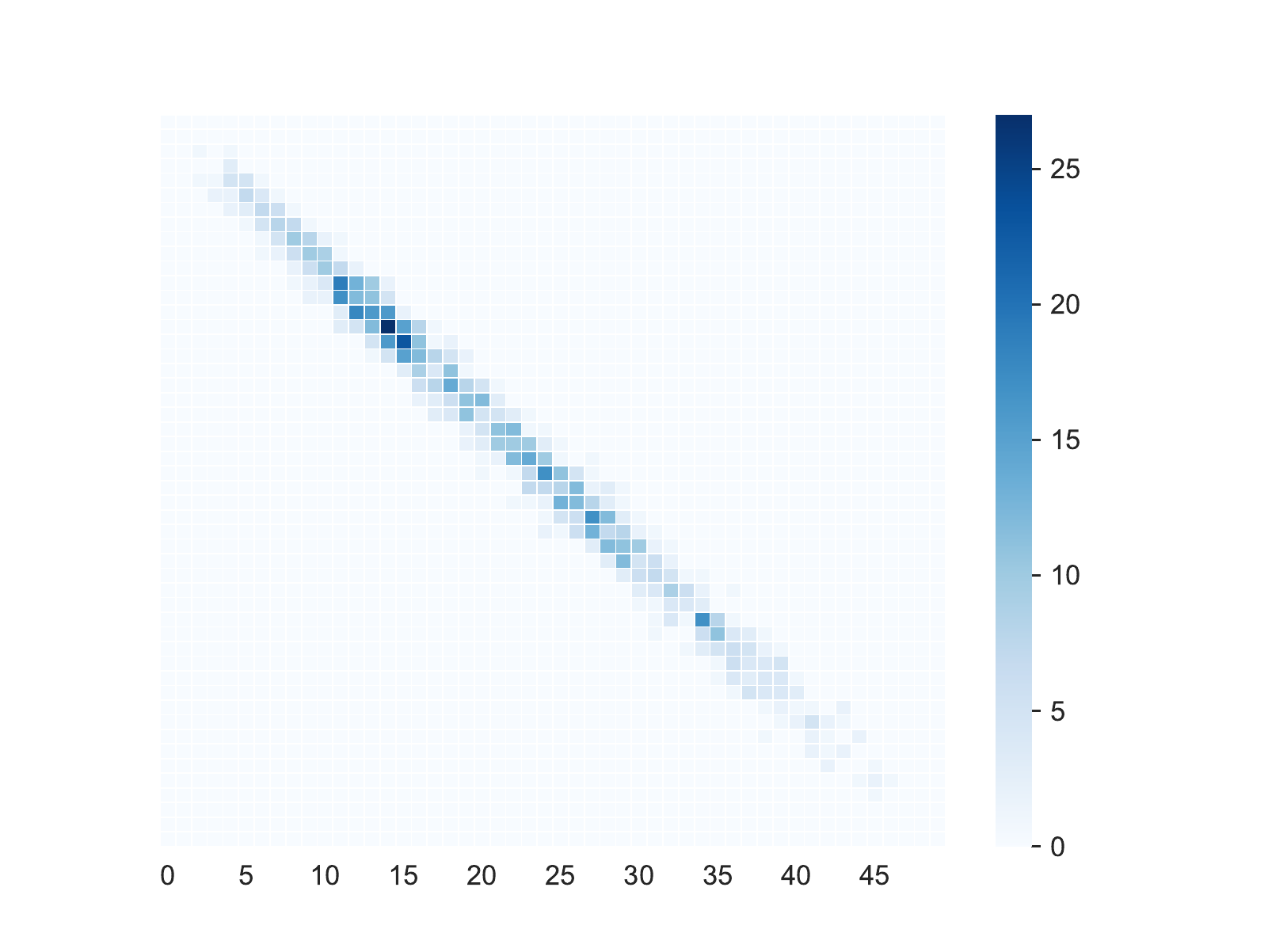}}\hspace{-3.4em}
\subfigure[Ornstein-Uhlenbeck]{\includegraphics[width=.531\textwidth]{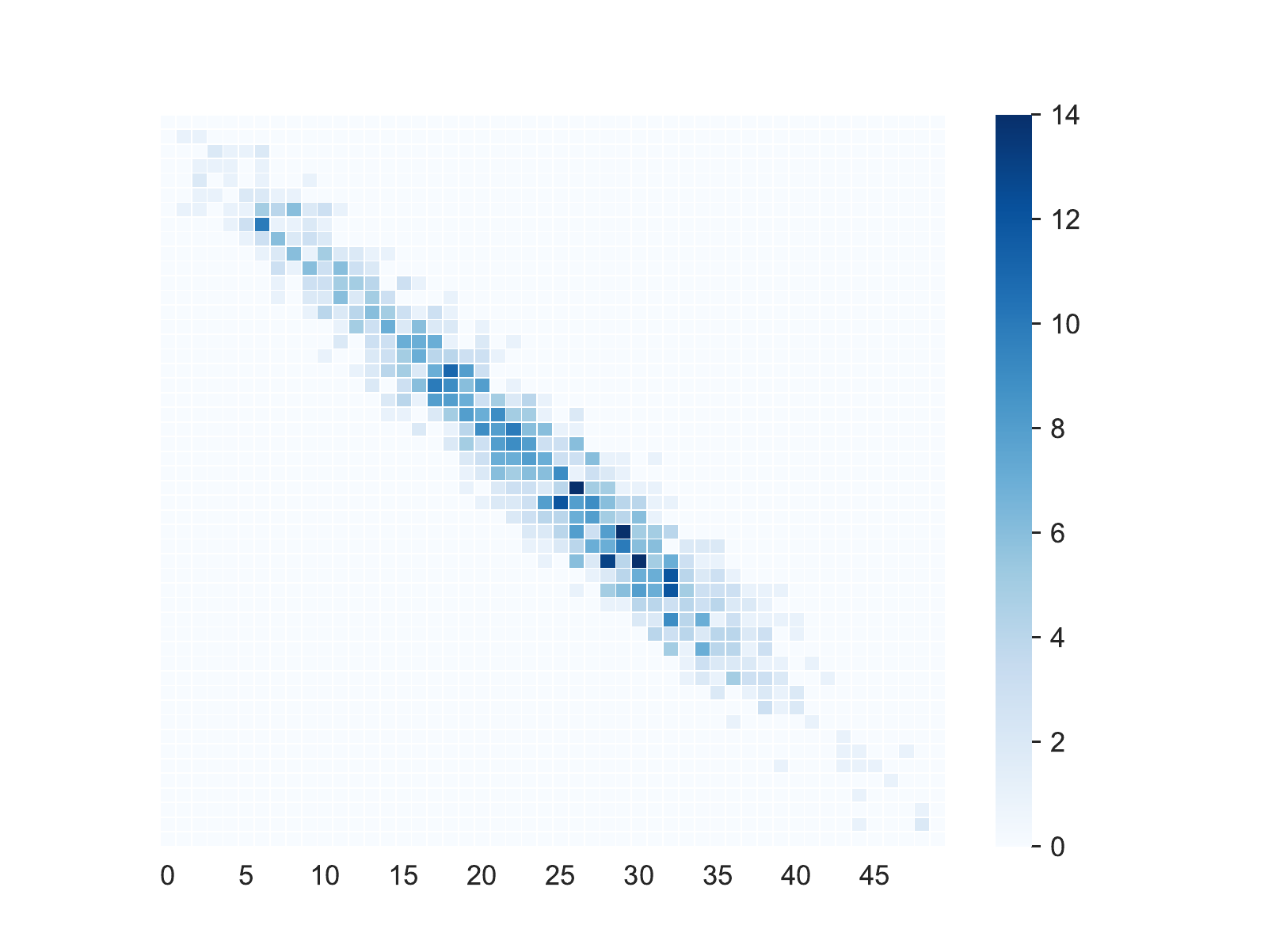}}
\caption{\small{Transition counts, $\mathbf C(\Delta)$ for two models, for a $m=50$ state CTMC, with $\Delta=1/250$, and $N=1250$. The x-axis denotes the state-space index, and the color bar shows the counts, $\mathbf C(\Delta)_{i,j}$ for each pair of states $i\rightarrow j$.
}}\label{fig:Transitions}
\end{figure}

We summarize the computational complexity of the CTMC-MLE algorithm in the next result. 
\begin{prop}[Computational Complexity] \label{Prop:CompC}
 Let $N_c$ denote the number of iterations for the MLE optimization to converge, where at each iteration a fixed number of likelihood evaluations are performed.  For CTMC-MLE, the cost is $\mathcal O(N + N_c\cdot( m^3 +  B\cdot m))$, where $B$ is the bandwidth of $\bm C(\Delta)$, that is $B:=\max_{i,j}\{|i-j|: \bm C(\Delta)_{i,j} > 0\}$, and typically $B << m$, as shown in Figure \ref{fig:Transitions}.   
\end{prop}
\begin{proof}
At initialization, the matrix $\bm C(\Delta)$ is pre-computed at a cost of $\mathcal O(N)$ (one pass through the sample), together with the bandwidth $B$.
The subsequent cost is driven the matrix exponential, $\bm T(\Delta)$, which can be computed with $\mathcal O(m^3)$ flops at each of the $N_c$ iterations. Computing the likelihood then costs only $B\cdot m$ at each iteration.
\end{proof}

From Proposition \ref{Prop:CompC}, there is an interesting computational advantage to the CTMC approximation over Exact MLE which grows with the sample size. The cost of Exact MLE is $\mathcal O(N_c\cdot C_p\cdot N)$, where  $C_p$ is the cost of evaluating the probability density (which can be significant in some cases, such as CIR).  By contrast, after initialization the CTMC method never revisits the sample (as it stores all sample information in the pre-computed $\bm C(\Delta)$), while Exact MLE requires a full pass back through the sample at each optimization step. The same argument holds in comparison with approximations such as Euler or the Shoji-Ozaki method.

\begin{remark} (Time-inhomogeneous diffusion)
In case the diffusion $S_t$ is time-inhomogeneous, that is, the dynamics of $S_t$ is described by
\begin{equation}
\label{Time-inhomogeneousSDE}
dS_t=\mu(t,S_t,\theta)dt+\sigma(t,S_t,\theta)dW_t, \quad 0\leq  t\geq T,
\end{equation}
then its corresponding infinitesimal generator  is given by
\begin{equation}
\label{time-inhomogeneousGen}
\mathcal L_t f(s)=\mu(t,s,\theta)\frac{\partial f}{\partial s}+
\frac{1}{2}\sigma^2(t,s,\theta)\frac{\partial^2f}{\partial s^2},\quad \forall f\in C^2_c(\mathbb S).
\end{equation}
We then build an approximating time-inhomogeneous CTMC with a generator
that is piecewise constant in time. More specifically, given a partition $\mathbb T=\{T_i\}_{i=0}^M$
with $T_0=0<T_1<\ldots<T_M=T$ of $[0,T]$,
let $\mathbf Q^{(j)}$ be an approximation of the infinitesimal generator $\mathcal L_{T_j}$. 
Then $S_t^m$ has a time-dependent generator given by
\begin{equation}
\label{time-dependentGen}
\mathbf Q_t=\sum_{i=1}^M \mathbf Q^{(j)}\mathbf I_{[T_{j-1},T_j)}(t),
\end{equation}
and all related quantities are defined analogously as the time-homogeneous case. 
For example,  the probability transition matrix
\begin{equation}
\mathbf T_t(\Delta)=\exp(\mathbf Q_t\Delta)=\sum_{i=0}^{\infty}\frac{(\mathbf Q_t\Delta)^i}{i!}.
\end{equation}
To keep the treatment focused, we will only consider the case of time-homogeneous diffusions, with the time-dependent case left as a natural extension.
\end{remark}

%%%%%%%%%%%%%%%%%%%%%%%%%%%%
\subsection{MLE Convergence Analysis}\label{sect:Theory}
%%%%%%%%%%%%%%%%%%%%%%%%%%%%
In this section,  we assume that the state space of $S_t$ is $[l,r]$\footnote{In case the domain of $S_t$
is, for example, of the form $(-\infty,+\infty)$ or $(0,\infty)$,
then we can choose $l,r$ such that $S_t\in [l,r]$
with high probability.} for $-\infty<l<r<\infty$.  Here we show that as the CTMC state space is refined ($m\rightarrow \infty$), the CTMC-MLE estimate will converge under reasonable regularity conditions to the Exact MLE estimate for any finite sample of size $N$. In the numerical experiments we will further demonstrate that accurate approximations, comparable with Exact MLE when it is available,  are obtained with a small (finite) number of states.
\begin{asp}
\label{Assumption-Coefficients}
$\mu(.,\theta)\in C^3([l,r]),\sigma^2(.,\theta)\in C^4([l,r])$.
\end{asp}
\begin{thm}
We have, under  Assumption \ref{Assumption-Coefficients},
$$
\widehat{\theta}_{N,m}\to\widehat{\theta}_N,  \quad\mbox{as}\quad m\to\infty.
$$
\end{thm}
\begin{proof}
First define
$$
\mathfrak{m}_{\theta}(s)=\frac{2}{\sigma^2(s,\theta)}\exp
\left(\int_l^x\frac{2\mu(y,\theta)}{\sigma^2(y,\theta)} dy\right).
$$
For $f\in L^2([l,r],\mathfrak{m}_{\theta})$,
consider the following PDE with the initial and boundary conditions
\begin{equation}
\label{PDE-equation}
\left\{
\begin{array}{ll}
\left(\mathcal L-\frac{\partial }{\partial t}\right)u(t,s)=0,\quad  t>0,\quad s\in (l,r),\\
u(t,l)=u(t,r)=0,\quad t\geq 0,\\
u(0,s)=f(s),\quad s\in (l,r).
\end{array}
\right.
\end{equation}
It is well-known (e.g, \cite{sheu1991some,downes2009bounds}) that
the transition density $p(t,s,y)$ is the  fundamental solution to the parabolic equation 
\eqref{PDE-equation}.
Next, consider the Sturm-Liouville eigenvalue problem
\begin{equation}
\label{PDE-equation1}
\left\{
\begin{array}{ll}
\mathcal L\phi(s)=\eta \phi(s), \quad s\in (l,r),\\
\phi(l)=\phi(r)=0.
\end{array}
\right.
\end{equation}
This is a regular Sturm-Liouville problem which
has a denumerable sequence of simple eigenvalues satisfying
$0<\eta_1<\eta_2<\ldots$. Also let the $\phi_i(s)$ be the normalized
eigenfunction corresponding to the eigenvalue $\eta_k$. Note that 
$\phi_i(s)$ and $\phi_j(s)$ are orthogonal for $i\neq j$. 
Moreover, $p_t(s,y)$ can be expressed as a bilinear eigenfunction expansion (\cite{mckean1956elementary,linetsky2007spectral})
\begin{equation}
p(t,s,y)=\sum_{i=1}^{\infty}e^{-\eta_i t}\phi_i(s)\phi_i(y)\mathfrak{m}_{\theta}(y).
\end{equation}

Without loss of generality, assuming  that the grid is uniform, that is $k=k_i=S_{i+1}-S_i,  \forall i$, note that
$k\to 0\Leftrightarrow m\to\infty$.
Recall from Theorem \ref{diagonalizable-theorem} that $-\mathbf Q$
has exactly $m$ distinct eigenvalues, which are denoted by
$0<\eta_{1,k}<\ldots<\eta_{m,k}$. Here we use $\eta_{i,k}$ to emphasize that the eigenvalues depend on 
the space step size $k$.
For each $i\in\{1,2,\ldots,m\}$, define
$$
\mathfrak{m}_{k,i}=\frac{2}{\sigma^2(l)}\prod_{j=1}^{i-1}\frac{\sigma^2(s_j)+\mu(s_j)k}{\sigma^2(s_{j+1})-\mu(s_{j+1})k},
$$
then we have (see example \cite{li2018error}),
$$
T(\Delta)_{i,j}=\sum_{l=1}^m e^{-\eta_{l,k}\Delta}\varphi_{l,i,k}\varphi_{l,j,k}\mathfrak{m}_{k,j},
$$
where in the above equation $\varphi_{i,k}=(\varphi_{1,i,k},\ldots,\varphi_{m,i,k})^t$
denote the eigenvector corresponding to the eigenvalue $\eta_{i,k}$.
In \cite[Theorem 3.1]{li2018error}, the authors show, under the Assumption \ref{Assumption-Coefficients}, that
there exists a constant $C_{\Delta}>0$ such that
\begin{align*}
\displaystyle\max_{1\leq i,j\leq m}|T(\Delta)_{i,j}-p(\Delta,s_i,s_j)|\leq C_{\Delta}k^2.
\end{align*}
Since $\ln(.)$ is a continuous function we have
\begin{align*}
\displaystyle\max_{1\leq i,j\leq m}|\ln T(\Delta)_{i,j}-\ln p(\Delta,s_i,s_j)|\to 0 \quad\mbox{as}\quad m\to\infty.
\end{align*}
As a result, we have
\begin{equation}
\sum_{0\leq i\leq N-1} \left(\ln T(\Delta)_{i,i+1}-\ln p(\Delta,s_i,s_{i+1})\right)\to 0\quad\mbox{as}\quad m\to \infty.
\end{equation}
Therefore $$
\widehat{\theta}_{N,m}\to\widehat{\theta}_N,  \quad\mbox{as}\quad m\to\infty.
$$
This completes the proof of the theorem.
\end{proof}

%%%%%%%%%%%%%%%%%
\subsection{Quasi-Newton Method}
%%%%%%%%%%%%%%%%%
We now describe a Quasi-Newton optimization approach for determining the approximate maximum likelihood solution with respect to the CTMC likelihood $L_{N,m}$. As demonstrated in Section \ref{sect:Theory}, the maximum likelihood estimate for the CTMC will converge to the true MLE as $m\rightarrow \infty$.
Let $\mathbf H$ be the $d\times d$ Hessian of $L_{N,m}$. That is,
$$
\mathbf H=(\mathbf H_{u,v}(\theta,\Delta))_{1\leq u,v\leq m}:=\left(\frac{\partial^2 L_{N,m}(\theta,\Delta)}{\partial \theta_u\partial\theta_v}
\right)_{1\leq u,v\leq d}.
$$
%Let $\mathbf H$ be the Hessian of $\mathbf L$, 
We will use a quasi-Newton method to approximate 
$\widehat{\theta}_{N,m}$, which follows the update rule:
\begin{equation}
\theta^{(k+1)}=\theta^{(k)}-[\mathbf H(\theta^{(k)})]^{-1}\nabla  L_{N,m}(\theta^{(k)}),
\end{equation}
where $\mathbf H$ is replaced with a suitable approximation.
The first term we estimate is the gradient, $\nabla  L_{N,m}$. 
To derive the gradient of $L_{N,m}$, we will require $\frac{\partial \mathbf Q}{\partial \theta_u}$.
Recall that $\theta=(\theta_1,\theta_2,\ldots,\theta_d)$. If we let $\partial_u\sigma=\partial\sigma/\partial\theta_u$, then by taking derivatives 
we have:
\begin{lem}\label{lem:PartQpartthet}
The derivatives of the generator defined in \eqref{q-MCA-Heston} with respect to model parameters are given by
$$
\frac{\partial \mathbf Q}{\partial \theta_u}=[q^\prime_{ij}]_{m\times m}, \quad 1\leq u \leq d,
$$
where 
\begin{equation}
\label{q-MCA-Heston}
q^\prime_{ij}=\left\{
\begin{array}{ll}
\displaystyle\frac{\partial_u\mu^-(s_i,\theta)}{k_{i-1}}+\displaystyle\frac{\partial_u\sigma^2(s_i,\theta)-(k_{i-1}\partial_u\mu^-(s_i,\theta)+k_i\partial_u\mu^+(s_i,\theta))}{k_{i-1}(k_{i-1}+k_i)},& \;\; \text{if}\;\; j=i-1,\\
\displaystyle\frac{\partial_u\mu^+(s_i,\theta)}{k_{i}}+\displaystyle\frac{\partial_u\sigma^2(s_i,\theta)-(k_{i-1}\partial_u\mu^-(s_i,\theta)+k_i\partial_u\mu^+(s_i,\theta))}{k_{i}(k_{i-1}+k_i)},& \;\; \text{if}\;\; j=i+1,\\
-q^\prime_{i,i-1}-q^\prime_{i,i+1},&\;\; \text{if}\;\;j = i,\\
0, & \;\; \text{if}\;\; j\neq i-1,i,i+1.\\
\end{array}
\right.
\end{equation}
\begin{proof}
The proof follows directly upon differentiating the terms in \eqref{q-MCA-Heston}.
\end{proof}
\end{lem}
Next define for $\Delta>0$ the $m\times m$ matrix
$\mathbf X(\mathbf\Lambda,\Delta)$, which is given by
\begin{equation}
\mathbf X(\mathbf \Lambda,\Delta)_{i,j}=
\left\{
\begin{array}{cc}
\Delta \exp(\Delta\lambda_i) &\mbox{if}\quad i=j,\\
\frac{\exp(\Delta\lambda_i)-\exp(\Delta\lambda_j)}{\lambda_i-\lambda_j}&\mbox{if}\quad i\neq j.
\end{array}
\right.
\end{equation}
\begin{lem}
The gradient $\nabla  L_{N,m}(\theta) \in \mathbb R^{d\times 1}$ is given by
\begin{equation}\label{eq:partLthe}
\frac{\partial L_{N,m}(\theta,\Delta)}{\partial \theta_u}=
\sum_{1\leq i,j \leq m}\left(\partial \mathbf Q/\partial \theta_u\circ \mathbf Z \right)_{i,j}, \quad 1\leq u\leq d,
\end{equation}
where
\begin{equation*}
\mathbf Z=\mathbf U((\mathbf V^T\mathbf D\mathbf U)\circ \mathbf X(\Lambda,\Delta))\mathbf V^T),
\end{equation*}
and $\mathbf D=(D_{i,j})$ is the $m\times m$ matrix defined by
$$
D_{i,j}=\mathbf C(\Delta)_{i,j}/\mathbf T_{i,j}, \quad 1\leq i, j \leq m.
$$
\end{lem}
\begin{proof}
This result 
follows from  \cite{kalbfleisch1985analysis, mcgibbon2015efficient}, which provides
\begin{equation}
\frac{\partial L_{N,m}(\theta,\Delta)}{\partial \theta_u}
=\sum_{1\leq i,j \leq m}\left(D\circ \mathbf V((\mathbf U^T(\partial \mathbf Q/\partial \theta_u)\mathbf V)\circ \mathbf X(\mathbf\Lambda,\Delta))\mathbf U^T \right)_{i,j}
\end{equation}
Equation \eqref{eq:partLthe} follows using the facts that $\sum_{i,j}(\mathbf A\circ \mathbf B)_{i,j}=Tr(\mathbf A\mathbf B^T)$
and $Tr(\mathbf A^T(\mathbf B\circ \mathbf C))=Tr(\mathbf B^T(\mathbf A\circ C))$.
\end{proof}
We note that the matrix $\mathbf Z$ is independent of $\theta$, hence can be computed beforehand.
This will reduce the cost of computing $\nabla  L_{N,m}$
substantially.

%%%%%%%%%%%%%%%%%%%%%%%%%%%%%%%%%%%
\subsection{Approximate Hessian }
Recall that,
$$
\mathbf H=(\mathbf H_{u,v}(\theta,\Delta))_{1\leq u,v\leq d}:=\left(\frac{\partial^2 L_{N,m}(\theta,\Delta)}{\partial \theta_u\partial\theta_v}
\right)_{1\leq u,v\leq d}.
$$
Using
\begin{align}
 L_{N,m}(\theta,\Delta)=\sum_{1\leq i,j\leq m}\mathbf C(\Delta)_{i,j}\ln \mathbf T(\Delta)_{i,j},
 \end{align}
it can be seen that the Hessian matrix is given by
\begin{equation}
\label{equation:Hessian}
\mathbf H_{u,v}(\theta,\Delta)=\sum_{i=1}^m\sum_{j=1}^m\mathbf C_{i,j}(\Delta)\left(\frac{\partial ^2\mathbf T_{i,j}/\partial \theta_u\partial \theta_v}{\mathbf T_{i,j}}
-\frac{(\partial \mathbf T_{i,j}/\partial \theta_u)(\partial \mathbf T_{i,j}/\partial \theta_v)}{\mathbf T^2_{i,j}}
 \right).
\end{equation}
From Lemma \ref{lem:PartQpartthet} and \cite{kalbfleisch1985analysis}, it follows that
\begin{align*}
\frac{\partial \mathbf T(\Delta)}{\partial\theta_u}&=
 \sum_{i=1}^{\infty}\frac{\partial}{\partial\theta_u}\left(\frac{(\mathbf Q\Delta)^i}{i!}\right)\\
 &=\sum_{i=1}^{\infty}\sum_{j=0}^{i-1} \mathbf Q^j \frac{\partial\mathbf Q}{\partial\theta_u}
 \mathbf Q^{i-1-l}\frac{\Delta^i}{i!}\\
 &=\sum_{i=1}^\infty\sum_{j=0}^{i-1}\bm V\bm \Gamma^j\bm U^T\frac{\partial\mathbf Q}{\partial\theta_u}
 \bm V \bm \Gamma^{i-1-j}\bm U^T\frac{\Delta^i}{i!}\\
 &=\bm V\left( \sum_{i=1}^\infty\sum_{j=0}^{i-1}\bm \Gamma^j\bm U^T\frac{\partial\mathbf Q}{\partial\theta_u}
 \bm V \bm \Gamma^{i-1-j}\frac{\Delta^i}{i!}\right)\bm U^T\\
 &=\mathbf V\left((\mathbf U^T\frac{\partial\mathbf Q}{\partial\theta_u}\mathbf V)\circ \mathbf X(\mathbf\Lambda,\Delta)\right)\mathbf U^T.
\end{align*}
From the equation \eqref{equation:Hessian} above, it can be seen that
direct calculation of the Hessian
is expensive since one
must compute the first derivative as well as the second 
derivative of $\mathbf T$. Following \cite{kalbfleisch1985analysis},
let $C_i(\Delta)=\sum_{j}\mathbf C_{i,j}(\Delta)$,
we can approximate $\mathbf C_{i,j}(\Delta)\approx\mathbf T_{i,j} C_i$.
And note that $\sum_{j}\partial ^2\mathbf T_{i,j}/\partial \theta_u\theta_v=0$, we can approximate
\begin{equation}
\mathbf H_{u,v}(\theta,\Delta)\approx-\sum_{i,j}\frac{C_i(\Delta)}{\mathbf T_{i,j}}
\frac{\partial \mathbf T_{i,j}}{\partial \theta_u}
\frac{\partial\mathbf T_{i,j}}{\partial \theta_v}=:{\mathbf {\widehat H}_{u,v}}(\theta,\Delta),
\quad 1\leq u,v\leq d.
\end{equation}
As a result, we use the following update
\begin{equation}
\theta^{(k+1)}=\theta^{(k)}-[\mathbf {\widehat H}(\theta^{(k)})]^{-1}\nabla  L_{N,m}(\theta^{(k)}).
\end{equation}

\section{Numerical Examples}\label{section:Numerical Examples}
%%%%%%%%%%%%%%%%%%%
This section provides various examples to demonstrate the CTMC-MLE framework. The first set of experiments aim to establish the closeness of CTMC-MLE to Exact MLE, in cases for which the exact transition density is known. We then consider examples for which approximations must be used, and we compare the CTMC-MLE method to the well-established approaches of Kessler \cite{kessler1997estimation} and Shoji-Ozaki \cite{shoji1998statistical}.

All experiments are conducted using Python 3.7.  To support future R\&D, we have developed an open source python library which includes the estimation procedures described below, for a wide variety of models. The library, called \textbf{pymle}, is freely available at: https://github.com/jkirkby3/pymle.

%%%%%%%%%%%%%%%%%%%
\subsection{Comparison to Exact MLE}
%%%%%%%%%%%%%%%%%%%
In these experiments we compare the CTMC-MLE estimator to Exact MLE, for several examples for which Exact MLE is feasible. We consider the Geometric Brownian Motion, Ornstein-Uhlenbeck, and Cox-Ingersoll-Ross processes in this section. The point of this comparison is to illustrate the closeness of CTMC-MLE to Exact MLE when it does exist, and towards this end we seek to control all other sources of variation in the estimation procedure.  Hence, the Exact MLE will be estimated by solving the numerical optimization $\widehat{\theta}_N:=\argmax_{\theta\in\Theta} L_N(\theta, \Delta)$, with $L_N(\theta, \Delta):=\sum_{n=1}^{N-1}\ln p(,\Delta,S_{n+1},S_{n})$, using a closed-form expression for $p(\Delta, s^\prime,s)$. 
\\
\\
\textbf{Experimental Design:} For each model, we consider three realistic estimation scenarios, each with $T=5$ years of data: 1) a sampling frequency of $52$ times per year, as is typical of weekly economic data, 2) a frequency of $250$, typical of daily financial market data, and 3) a frequency of $1000$, to represent a high-frequency intra-day sampling scenario.
 For each trial (which we repeat 500 times for each experiment), we simulate a sample path which is used for both CTMC and Exact MLE, and we apply the same estimation procedure for both (using a constrained trust-region solver).  When exact simulation is unavailable, we use the Milstein scheme which we combine with sub-stepping at a rate of 10 additional steps between each sample point to reduce bias in the simulated trajectories.

%%%%%%%%%%%%%%%%%%%
\subsubsection{Geometric Brownian Motion (GBM)}
%%%%%%%%%%%%%%%%%%%
In the first example, we compare the CTMC-MLE with Exact MLE for the GBM process, with parameters $\mu\in \mathbb R, \sigma >0$, and dynamics given by
$$
dS_t= S_t\mu dt+S_t \sigma dW_t.
$$ 
This model is widely used in economics and finance to model the dynamics of a risky asset. For any $t'>t\geq 0$,  $\Delta:=t'-t$, and $\theta:=(\mu, \sigma)$,  the log-normal transition density is given in closed form by
\[
p(\Delta, s^\prime,s)=\frac{1}{s' \sigma_\Delta\sqrt{2\pi}}\exp\left(-\frac{(\ln(s') - \mu_\Delta(s))^2}{2\sigma_\Delta^2} \right)
\]
where $\mu_\Delta(s):=\ln(s) + \left(\mu - \frac{1}{2}\sigma^2\right)\Delta$, and $\sigma_\Delta:=\sigma \sqrt{\Delta}$.

In Table \ref{table:GBMExactVC} we compare the two estimators.  For GBM, both estimators have some (comparable) difficulty estimating the drift (in small samples), but we can see that the diffusion parameter is very accurately estimated by both, with low standard deviation. In particular, the exact and CTMC estimates are nearly identical (in terms of error and standard deviation).  Overall, the CTMC-MLE approximates Exact MLE very well.

%\begin{center}
%%\FloatBarrier
%\begin{table}[h!t!b!]
%\centering
%  \addtolength{\tabcolsep}{1pt}
%\scalebox{.9}{
%\begin{tabular}{llccccccc}
%    %\multicolumn{9}{c}{European Call}\\\hline
%    \hline
%\toprule
%  & & \multicolumn{1}{c}{}&\multicolumn{3}{c}{CTMC MLE} & \multicolumn{3}{c}{Exact MLE}\\
%    \cmidrule(r{1em}){4-6}\cmidrule{7-9}
%  $N$ & $1/\Delta$ & True Param. &$\hat \theta_{N,m}$& $\hat \theta_{N,m} - \theta $ & $\text{sd}(\hat \theta_{N,m})$&  $\hat \theta_{N}$& $\hat \theta_{N} - \theta $ & $\text{sd}(\hat \theta_{N} )$ \\
%    \toprule
% 260 & 52 & $\mu=0.030$   & -0.005 & 0.035 & 0.052    &    -0.006 & 0.036 & 0.053 \\ 
% & & $\sigma=0.150$   & 0.151 & -0.001 & 0.004    &    0.151 & -0.001 & 0.004 \\  \hline 
% 1250 & 250 & $\mu=0.030$   & 0.019 & 0.011 & 0.036    &    0.022 & 0.008 & 0.032 \\ 
% & & $\sigma=0.150$   & 0.152 & -0.002 & 0.003    &    0.152 & -0.002 & 0.003 \\\hline
% 5000 & 1000 & $\mu=0.030$   & 0.002 & 0.028 & 0.032    &    0.005 & 0.025 & 0.037 \\ 
% & & $\sigma=0.150$   & 0.155 & -0.005 & 0.003    &    0.154 & -0.004 & 0.002 \\    
%    \bottomrule
%\end{tabular}}
%\caption{\small{GBM - comparison of CTMC vs Exact MLE, with $m=300$ states. Results from 50 repeated simulations, with the same randomized initial guess. Initial $S_0=100$.  Fixed time horizon $T=5$ with varying sampling frequency, $1/\Delta$.}} % title of Table
%  \label{table:GBMExactVC}
%\end{table}
%\FloatBarrier
%\end{center}

\begin{center}
%\FloatBarrier
\begin{table}[h!t!b!]
\centering
  \addtolength{\tabcolsep}{1pt}
\scalebox{.9}{
\begin{tabular}{llccccccc}
    %\multicolumn{9}{c}{European Call}\\\hline
    \hline
\toprule
  & & \multicolumn{1}{c}{}&\multicolumn{3}{c}{CTMC-MLE} & \multicolumn{3}{c}{Exact MLE}\\
    \cmidrule(r{1em}){4-6}\cmidrule{7-9}
  $N$ & $1/\Delta$ & True Param. &$\hat \theta_{N,m}$& $\hat \theta_{N,m} - \theta $ & $\text{sd}(\hat \theta_{N,m})$&  $\hat \theta_{N}$& $\hat \theta_{N} - \theta $ & $\text{sd}(\hat \theta_{N} )$ \\
    \toprule
 260 & 52 & $\mu=0.030$  & 0.025 & 0.005 & 0.057  & 0.027 & 0.003 & 0.056  \\ 
  & & $\sigma=0.150$   & 0.150 & 0.000 & 0.007  & 0.150 & 0.000 & 0.007  \\ \hline
 1250 & 250 & $\mu=0.030$  & 0.021 & 0.008 & 0.059  & 0.023 & 0.007 & 0.059  \\ 
  & & $\sigma=0.150$   & 0.150 & 0.000 & 0.003  & 0.150 & 0.000 & 0.003  \\ \hline
 5000 & 1000 & $\mu=0.030$  & 0.024 & 0.006 & 0.054  & 0.024 & 0.006 & 0.054  \\ 
  & & $\sigma=0.150$   & 0.150 & -0.000 & 0.002  & 0.150 & -0.000 & 0.002  \\ 
    \bottomrule
\end{tabular}}
\caption{\small{GBM - comparison of CTMC vs Exact MLE, with $m=300$ states. Results from 500 repeated simulations, with the same randomized initial guess. Initial $S_0=100$.  Fixed time horizon $T=5$ with varying sampling frequency, $1/\Delta$.}} % title of Table
  \label{table:GBMExactVC}
\end{table}
\FloatBarrier
\end{center}

\begin{figure}[h!t!b!]
\centering     %%% not \center
\subfigure{\includegraphics[width=.51\textwidth]{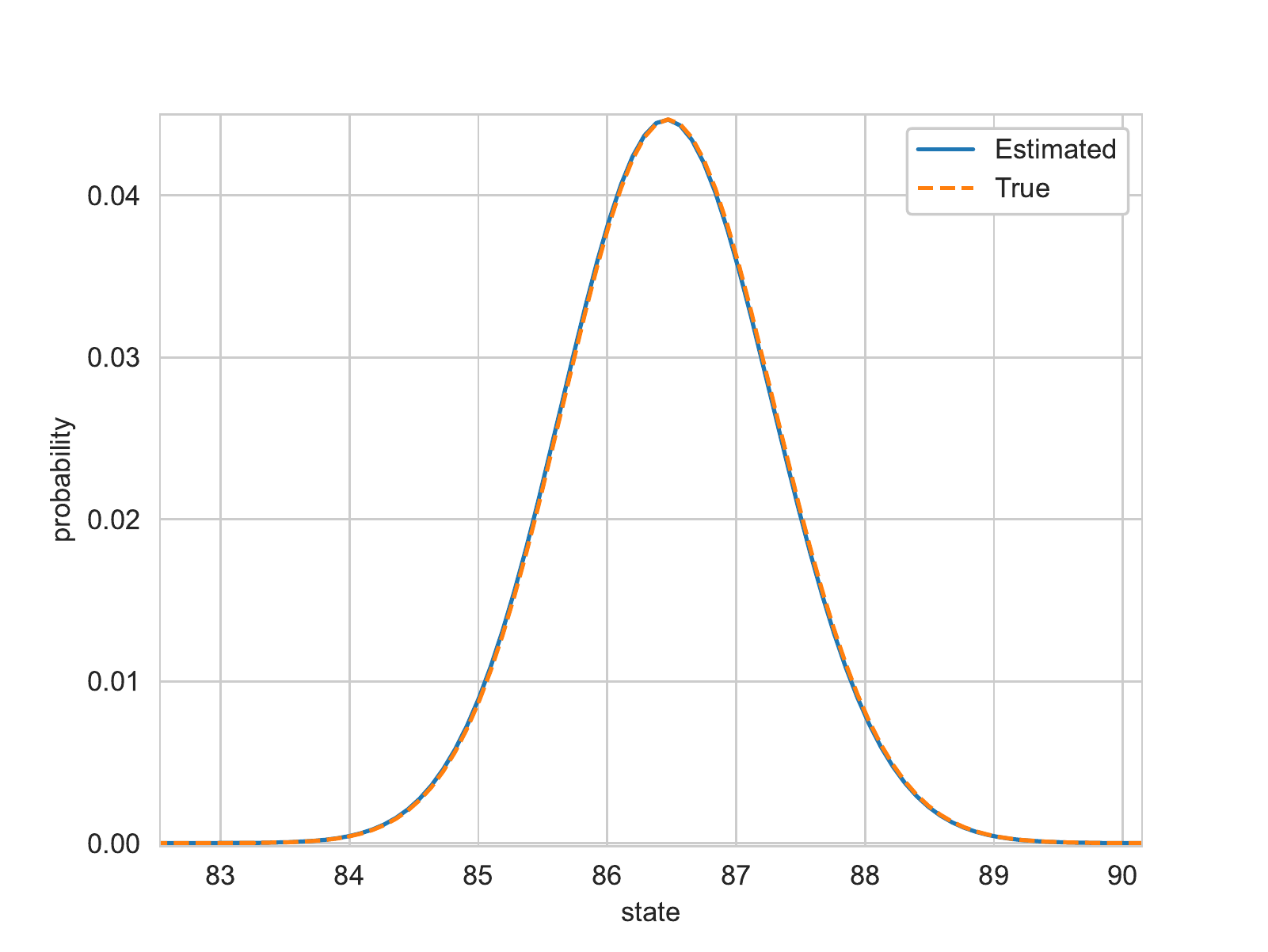}}\hspace{-1.8em}
\subfigure{\includegraphics[width=.51\textwidth]{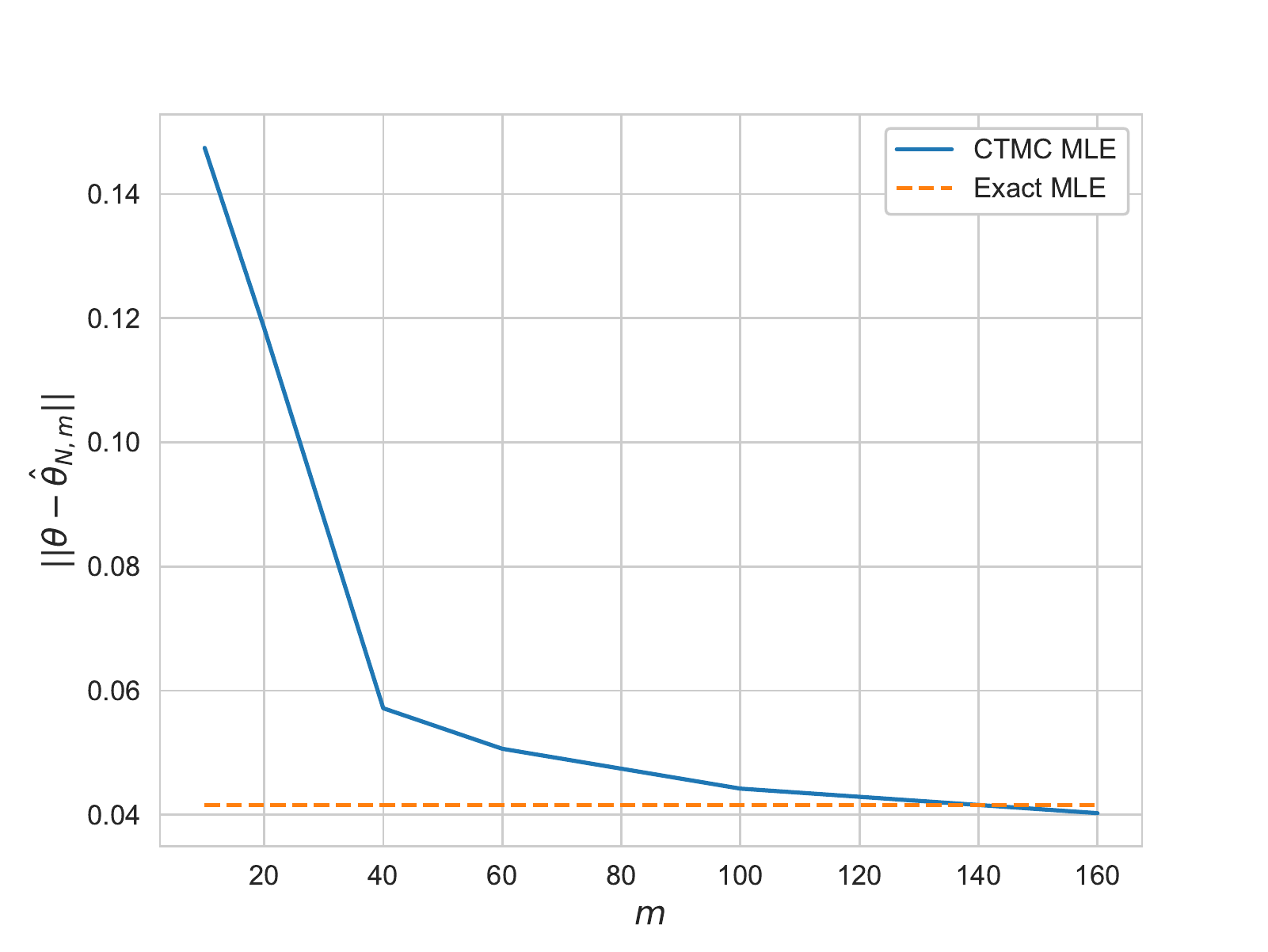}}
\caption{\small{GBM Example. Left: Transition density, $\bm T(\Delta)$ with $m=300$ state CTMC, and estimate $\hat \theta_{N,m}$. Right: convergence of $\hat \theta_{N,m}$ as function of states, $m$. Params: $\Delta=1/250$, and $N=1250$. 
}}\label{fig:GBMexample}
\end{figure}

%\begin{figure}[h!t!b!]
%\centering     %%% not \center
%\includegraphics[width=.7\textwidth]{./images/GBM_TransitionDensity.eps}
%\caption{\small{
%}}\label{fig:GBMDens}
%\end{figure}

In the left panel of Figure \ref{fig:GBMexample}, we show how the estimated CTMC (with $\hat \theta_{N,m}$) transition density approximates the true transition density for the case of $1/\Delta = 250$. 
The right panel demonstrates the fast convergence of the CTMC approximation, which we measure using the $l^2$ norm $||\theta - \hat\theta_{N,m}||:=\frac{1}{R}\sum_{r=1}^R\left(\sum_{i=1}^d(\theta(i) - \hat\theta^r_{N,m}(i))^2 \right)^{1/2}$, which is averaged over $R=500$ replications. We denote by $\hat\theta^r_{N,m}(i)$ the estimate of parameter $i\in\{1,\ldots, d\}$ for the $r^{th}$ replication of the experiment. In the numerical experiments, we utilize $m=300$ states in the approximating CTMC.

%%%%%%%%%%%%%%%%%%%
\subsubsection{Ornstein-Uhlenbeck (OU)}
%%%%%%%%%%%%%%%%%%%
In the second example we consider the mean-reverting OU model, with
dynamics ($\sigma>0$, $\kappa,\mu\in \mathbb R$)
$$
dS_t=\kappa(\mu-S_t)dt+\sigma dW_t.
$$
Among its numerous applications, OU is commonly used to model the instantaneous short interest rate (\cite{vasicek1977equilibrium}) in economics, as well as commodity prices \cite{schwartz1997stochastic}.
Its true transition density $p(s^\prime|s;\theta)$ is Gaussian,
\[
p(\Delta, s^\prime,s)= \frac{1}{\sigma_\Delta\sqrt{2\pi}}\exp\left(-\frac{(s' - \mu_\Delta(s))^2}{2\sigma_\Delta^2} \right),
\]
with mean $\mu_\Delta(s):=\mu +(s-\mu)e^{-\kappa\Delta}$, and variance $\sigma^2_\Delta:=\frac{\sigma^2}{2\kappa}(1-e^{-2\kappa \Delta})$, allowing us to determine the Exact MLE.
\\

Table \ref{table:OUExactVC} compares CTMC-MLE with Exact MLE, and the result is nearly an identical match between the two estimators. Without using Python vectorization for either algorithm, Exact MLE requires about 45 seconds on average for this example (when $1/\Delta=1000$), compared with less than 15 seconds for CTMC-MLE.\footnote{Vectorization can be used for languages such as Python to speed up the computation of Exact MLE significantly. For this comparison, we avoid using such techniques as they obscure the true computational cost for comparison purposes.}

\begin{center}
%\FloatBarrier
\begin{table}[h!t!b!]
\centering
  \addtolength{\tabcolsep}{1pt}
\scalebox{.9}{
\begin{tabular}{llccccccc}
    %\multicolumn{9}{c}{European Call}\\\hline
    \hline
\toprule
 & &  \multicolumn{1}{c}{}&\multicolumn{3}{c}{CTMC-MLE} & \multicolumn{3}{c}{Exact MLE}\\
   \cmidrule(r{1em}){4-6}\cmidrule{7-9}
 $N$ & $1/\Delta$ &  True Param. &$\hat \theta_{N,m}$& $\hat \theta_{N,m} - \theta $ & $\text{sd}(\hat \theta_{N,m})$&  $\hat \theta_{N}$& $\hat \theta_{N} - \theta $ & $\text{sd}(\hat \theta_{N})$ \\
    \toprule
 260 & 52 & $\kappa=4.000$  & 4.604 & -0.604 & 1.079  & 4.706 & -0.706 & 1.087  \\ 
  & & $\mu=0.200$   & 0.198 & 0.002 & 0.046  & 0.202 & -0.002 & 0.046  \\ 
  & & $\sigma=0.400$   & 0.399 & 0.001 & 0.018  & 0.404 & -0.004 & 0.018  \\   \hline 
%%%%%%%%%%%%%%%%%%%%%%%%%%%%%%%
 1250 & 250 & $\kappa=4.000$  & 4.690 & -0.690 & 1.105  & 4.710 & -0.711 & 1.099  \\ 
  & & $\mu=0.200$   & 0.198 & 0.002 & 0.046  & 0.201 & -0.001 & 0.046  \\ 
  & & $\sigma=0.400$   & 0.399 & 0.001 & 0.008  & 0.400 & -0.000 & 0.008  \\ \hline
%%%%%%%%%%%%%%%%%%%%%%%%%%%%%%% 
 5000 & 1000 & $\kappa=4.000$  & 4.632 & -0.632 & 1.080  & 4.617 & -0.617 & 1.086  \\ 
  & & $\mu=0.200$   & 0.200 & -0.000 & 0.043  & 0.202 & -0.002 & 0.043  \\ 
  & & $\sigma=0.400$   & 0.401 & -0.001 & 0.004  & 0.400 & 0.000 & 0.004  \\ 
    \bottomrule
\end{tabular}}
\caption{\small{OU - comparison of CTMC-MLE vs Exact MLE, with $m=300$ states. Results from 500 repeated simulations, with the same randomized initial guess. Initial $S_0=0.2$. Fixed time horizon $T=5$ with varying sampling frequency, $1/\Delta$.}} % title of Table
  \label{table:OUExactVC}
\end{table}
\FloatBarrier
\end{center}

%%%%%%%%%%%%%%%%%%%
\subsubsection{Cox-Ingersoll-Ross (CIR)}
%%%%%%%%%%%%%%%%%%%
The dynamics of $S_t$ under CIR is given by $(\kappa, \mu, \sigma > 0$)
$$
dS_t=\kappa(\mu-S_t)dt+\sigma\sqrt{S_t} dW_t.
$$
It can be shown that $S_t\geq 0$ almost surely, and the CIR model is
widely used to model the short term interest rates (\cite{cox2005theory}) or equity volatilities (\cite{heston1993closed}). The true transition density function is given by
$$
p(\Delta, s^\prime,s)=
\frac{e^{\kappa \Delta}}{2c(\Delta)}
\left(\frac{s^\prime e^{\kappa \Delta}}{s} \right)^{(d-2)/4}
\exp\left(-\frac{s+s^\prime e^{\kappa \Delta}}{2c(\Delta)} \right)
I_{d/2-1}\left(\frac{\sqrt{s s^\prime e^{-\kappa \Delta}}}{c(\Delta)} \right),
$$
where
$$
c(\Delta)=\frac{\sigma^2}{4\kappa}(e^{\kappa \Delta}-1), \quad d=\frac{4\kappa \mu}{\sigma^2},
$$
and 
$$
I_{\gamma}(x)=\sum_{i=0}^{\infty}
\frac{(x/2)^{2i+\gamma}}{i!\Gamma(i+\gamma+1)}
$$
is the modified Bessel function of the first kind of order $\gamma$. Numerical evaluation of 
$p(s^\prime|s; \theta)$ is delicate, and is best implemented using the exponentially damped Bessel function.

Table \ref{table:CIRExactVC} summarizes the results for CTMC-MLE, and the estimates obtained from Exact MLE.  As with the GBM and OU examples, CTMC-MLE provides very similar estimates as Exact MLE for the CIR model.
 In Figure \ref{fig:ReSim} we illustrate qualitative similarity between the estimated model and the true (unknown) process for CIR (Left). After simulating the sample trajectory and estimating the coefficients, we re-simulate the process using the same seed as the original sample but with the estimated coefficients. 

\begin{figure}[h!t!b!]
\centering     %%% not \center
\subfigure{\includegraphics[width=.51\textwidth]{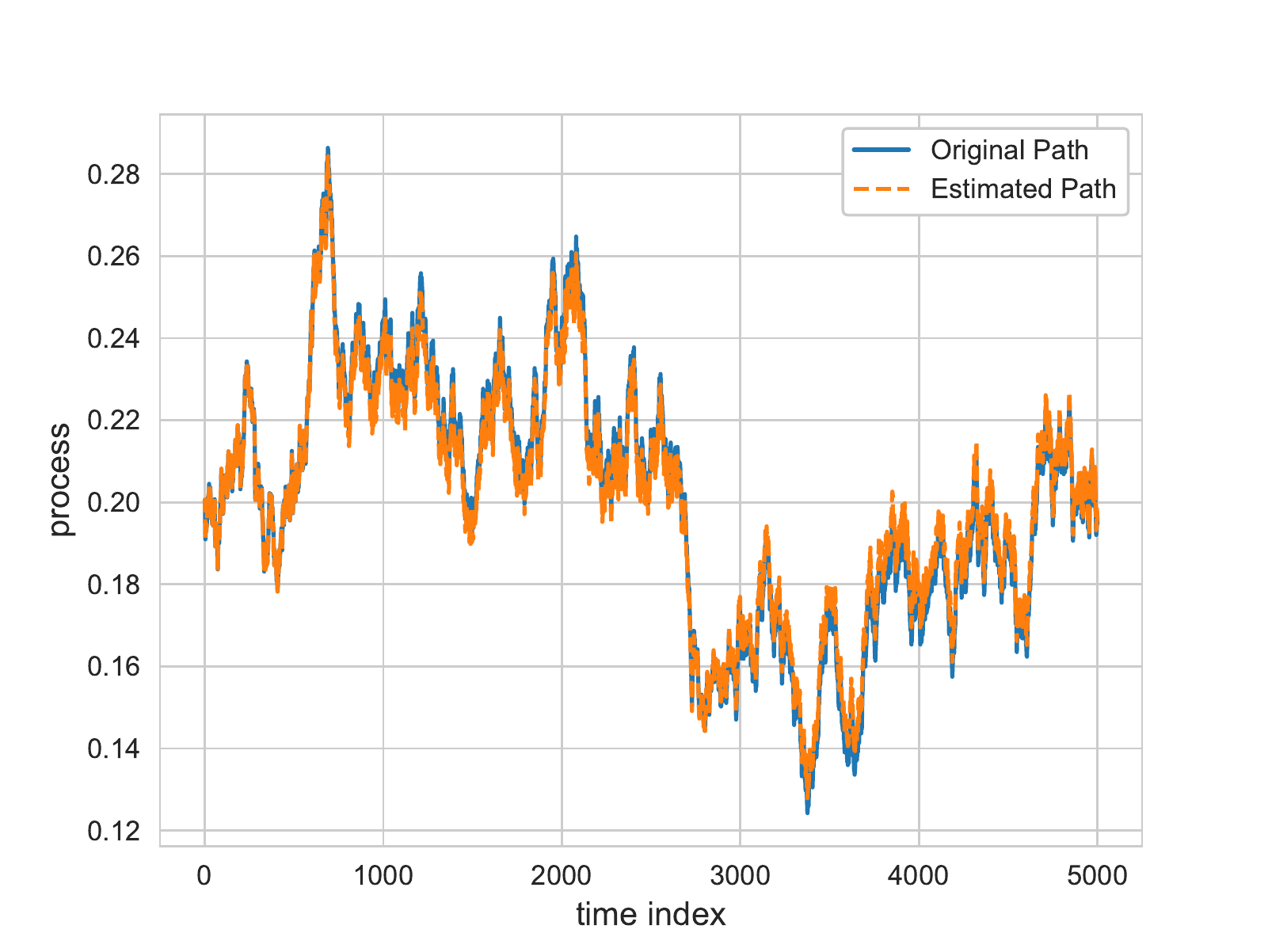}}\hspace{-1.8em}
\subfigure{\includegraphics[width=.51\textwidth]{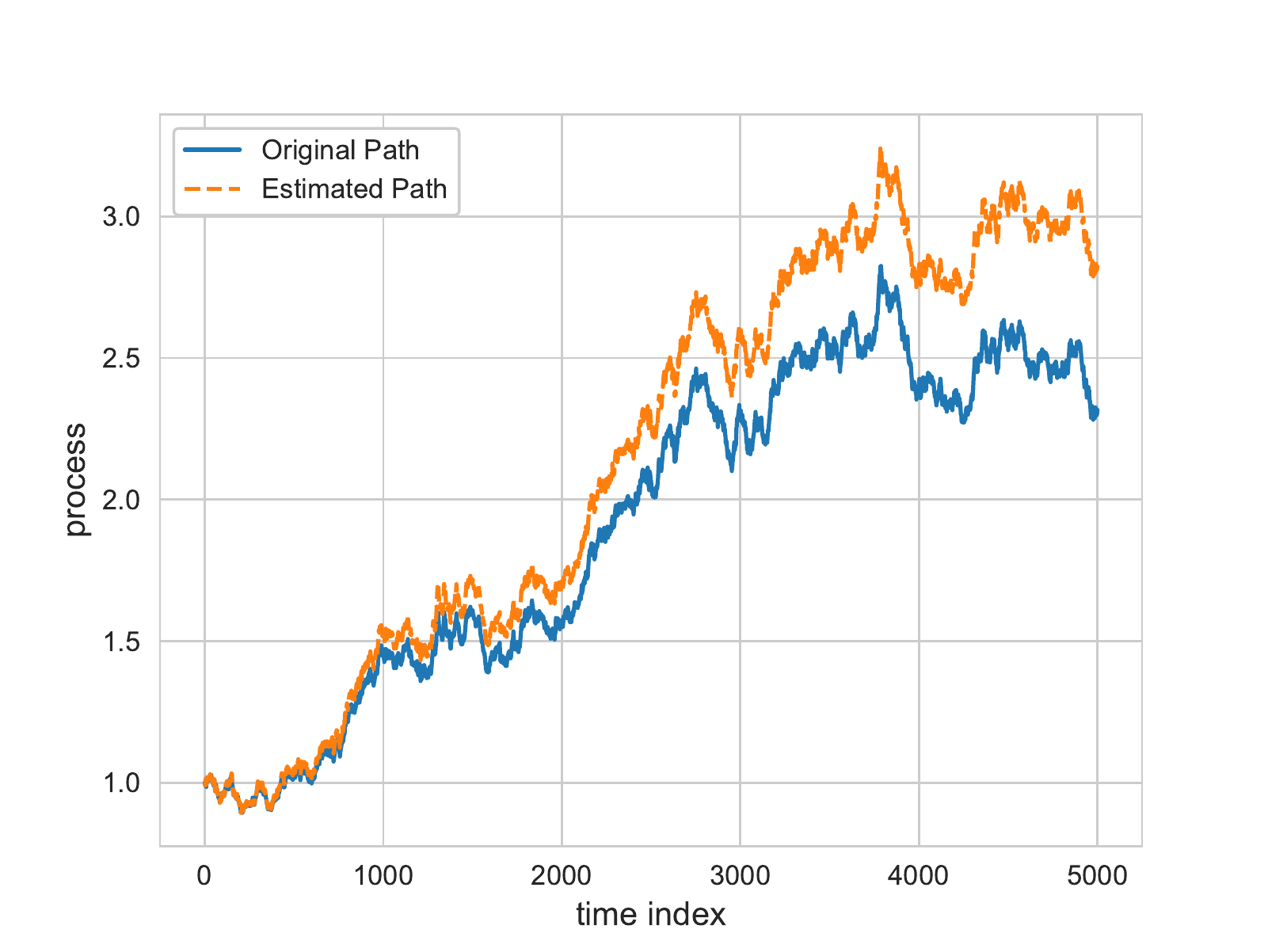}}
\caption{\small{Estimated re-simulation example. Comparison of the sample path used for estimation, and the re-simulated path using CTMC-MLE estimated parameters. Left: CIR, with estimated params (2.313, 0.201, 0.149). Right: CKLS with estimated params (0.0132, 0.1342 0.2139, 0.5410) and the same initial seed. $\Delta=1/1000$, and $N=5000$. 
}}\label{fig:ReSim}
\end{figure}

\begin{center}
%\FloatBarrier
\begin{table}[h!t!b!]
\centering
  \addtolength{\tabcolsep}{1pt}
\scalebox{.9}{
\begin{tabular}{llccccccc}
    %\multicolumn{9}{c}{European Call}\\\hline
    \hline
\toprule
 & &  \multicolumn{1}{c}{}&\multicolumn{3}{c}{CTMC-MLE} & \multicolumn{3}{c}{Exact MLE}\\
   \cmidrule(r{1em}){4-6}\cmidrule{7-9}
 $N$ & $1/\Delta$ &  True Param. &$\hat \theta_{N,m}$& $\hat \theta_{N,m} - \theta $ & $\text{sd}(\hat \theta_{N,m})$&  $\hat \theta_{N}$& $\hat \theta_{N} - \theta $ & $\text{sd}(\hat \theta_{N})$ \\
    \toprule
 260 & 52 & $\kappa=2.000$  & 2.618 & 0.618 & 1.072  & 2.493 & 0.493 & 0.995  \\ 
  & & $\mu=0.200$   & 0.198 & -0.002 & 0.014  & 0.187 & -0.013 & 0.018  \\ 
  & & $\sigma=0.150$   & 0.150 & 0.000 & 0.007  & 0.189 & 0.039 & 0.051  \\  \hline
 1250 & 250 & $\kappa=2.000$  & 2.759 & 0.759 & 1.111  & 2.735 & 0.735 & 1.094  \\ 
  & & $\mu=0.200$   & 0.199 & -0.001 & 0.015  & 0.200 & 0.000 & 0.015  \\ 
  & & $\sigma=0.150$   & 0.150 & 0.000 & 0.003  & 0.150 & 0.000 & 0.003  \\ \hline
 5000 & 1000 & $\kappa=2.000$  & 2.736 & 0.736 & 1.228  & 2.579 & 0.579 & 1.161  \\ 
  & & $\mu=0.200$   & 0.200 & 0.000 & 0.015  & 0.200 & 0.000 & 0.019  \\ 
  & & $\sigma=0.150$   & 0.151 & 0.001 & 0.002  & 0.150 & 0.000 & 0.002  \\ 
%%%%%%%%%%%%%%%%%%%%%%%%%%%%%%%
%%%%%%%%%%%%%%%%%%%%%%%%%%%%%%% 
    \bottomrule
\end{tabular}}
\caption{\small{CIR - comparison of CTMC vs Exact MLE, with $m=300$ states. Results from 500 repeated simulations, with the same randomized initial guess. Initial $S_0=0.15$. Fixed time horizon $T=5$ with varying sampling frequency, $1/\Delta$.}} % title of Table
  \label{table:CIRExactVC}
\end{table}
\FloatBarrier
\end{center}

%
%%%%%%%%%%%%%%%%%%%%
%\noindent\textbf{Ornstein-Uhlenbeck (OU) model}: 
%%%%%%%%%%%%%%%%%%%%
%In this model, the
%dynamics of $S_t$ is given by
%$$
%dS_t=\kappa(\mu-S_t)dt+\sigma dW_t.
%$$
%This OU is usually used to model short time interest rate (\cite{vasicek1977equilibrium}).
%Its true transition density $p(s^\prime, t^\prime|s,t)$ is normally distributed 
%with mean $\mu +(s-\mu)e^{-\kappa(t^\prime-t)}$, and variance $\frac{\sigma^2}{2\kappa}(1-e^{-2\kappa(t^\prime-t)})$, allowing us to perform exact maximum likelihood estimation.\red{Tested, works well}
%\\

%%%%%%%%%%%%%%%%%%%
\subsection{Comparison with Psuedo-Likelihood}
%%%%%%%%%%%%%%%%%%%
Except for a handful of special cases, Exact MLE is unavailable, and some form of approximation is required. To obtain benchmarks, we utilize various ``Psuedo-Likelihood" approaches based on approximations of the SDE. For example, using an Euler approximation of the SDE, yields the approximate density
\[
p(\Delta, s^\prime,s)\approx \frac{1}{\sqrt{2\pi \Delta \sigma^2(s,\theta)}}\exp\left(-\frac{(s'-s - \mu(s,\theta)\Delta)^2}{2\Delta \sigma^2(s,\theta)} \right).
\]
Euler's approximation only works well for very small $\Delta$, so we also
compare against the more accurate methods of Kessler \cite{kessler1997estimation} and Shoji-Ozaki \cite{shoji1998statistical}.\footnote{The Elerian method was also tested, based on a Milstein approximation to the SDE, but we found some numerical instabilities with this approach.}

%%%%%%%%%%%%%%%%%%%
\subsubsection{Chan-Karolyi-Longstaff-Sanders (CKLS)}\label{sect:CKLS}
%%%%%%%%%%%%%%%%%%%
Another interesting example is the Chan-Karolyi-Longstaff-Sanders (CKLS) family of models (see \cite{chan1992empirical}), which is a four-parameter extension of the CEV model given by
$$
dS_t = (\theta_1 + \theta_2 S_t)dt + \theta_3 S_t^{\theta_4} dW_t.
$$
This model does not admit an explicit transition density, except in the case where $\theta_1=0$ or $\theta_4=1/2$.  We assume that $\theta_3>0$, and the process is positive as long as $\theta_1, \theta_2 >0$ and $\theta_4 > 1/2$.  Table \ref{table:CKLS} compares the CTMC-MLE method with that of Kessler and Shoji-Ozaki for the CKLS example.  All methods perform comparably well, and we notice that similar to the GBM case, the drift parameters, $\theta_1,\theta_2$, are considerably harder to estimate than the diffusion parameters, $\theta_3, \theta_4$.  The right panel of Figure \ref{fig:ReSim} illustrates this difficulty,  where we notice that the diffusive characteristics of the true and re-simulated path are very similar, but we over-estimate $\theta_2$, which is the portion of the drift that depends on $S_t$.   The CTMC -MLE method performs comparatively well at estimating the parameters of the diffusion component, namely $\theta_3$ and $\theta_4$.  Across the numerical experiments, this is a recurring phenomenon, even when compared to Exact MLE, and it is most notable for lower sampling frequencies. 

\begin{center}
%\FloatBarrier
\begin{table}[h!t!b!]
\centering
  \addtolength{\tabcolsep}{1pt}
\scalebox{.9}{
\begin{tabular}{llccccccc}
    %\multicolumn{9}{c}{European Call}\\\hline
    \hline
\toprule
 & &  \multicolumn{1}{c}{}&\multicolumn{2}{c}{CTMC-MLE} & \multicolumn{2}{c}{Kessler} & \multicolumn{2}{c}{Shoji-Ozaki}\\
   \cmidrule(r{1em}){4-5}\cmidrule(r{1em}){6-7}\cmidrule{8-9}
 $N$ & $1/\Delta$ &  True Param. & $\hat \theta_{N,m} - \theta $ & $\text{sd}(\hat \theta_{N,m})$& $\hat \theta_{N} - \theta $ & $\text{sd}(\hat \theta_{N})$ & $\hat \theta_{N} - \theta $ & $\text{sd}(\hat \theta_{N})$ \\
    \toprule
  120 & 24 & $\theta_1=0.010$  &  0.117 & 0.113  &  0.108 & 0.106   &  0.123 & 0.115  \\ 
  & & $\theta_2=0.100$   &  -0.079 & 0.050  &  -0.057 & 0.051   &  -0.071 & 0.058  \\ 
  & & $\theta_3=0.200$   &  -0.006 & 0.018  &  -0.008 & 0.015   &  -0.006 & 0.019  \\ 
  & & $\theta_4=0.600$   &  0.105 & 0.289  &  0.161 & 0.143   &  0.132 & 0.301  \\ \hline
 260 & 52 & $\theta_1=0.010$  &  0.113 & 0.112  &  0.042 & 0.048   &  0.124 & 0.112  \\ 
  & & $\theta_2=0.100$   &  -0.083 & 0.048  &  -0.031 & 0.055   &  -0.076 & 0.054  \\ 
  & & $\theta_3=0.200$   &  -0.002 & 0.013  &  0.014 & 0.031   &  -0.002 & 0.013  \\ 
  & & $\theta_4=0.600$   &  0.043 & 0.196  &  0.160 & 0.102   &  0.074 & 0.204  \\ \hline
 1250 & 250 & $\theta_1=0.010$  &  0.123 & 0.119  &  0.096 & 0.101   &  0.124 & 0.119  \\ 
  & & $\theta_2=0.100$   &  -0.083 & 0.051  &  -0.060 & 0.054   &  -0.079 & 0.055  \\ 
  & & $\theta_3=0.200$   &  0.001 & 0.007  &  -0.001 & 0.011   &  -0.001 & 0.007  \\ 
  & & $\theta_4=0.600$   &  -0.001 & 0.100  &  0.042 & 0.101   &  0.014 & 0.100  \\ 
    \bottomrule
\end{tabular}}
\caption{\small{CKLS - comparison of CTMC-MLE vs Euler and Shoji-Ozaki, with $m=300$ states. Results from 500 repeated simulations, with the same randomized initial guess. Initial $S_0=1.0$. Fixed time horizon $T=5$ with varying sampling frequency, $1/\Delta$.}} % title of Table
  \label{table:CKLS}
\end{table}
\FloatBarrier
\end{center}

%%%%%%%%%%%%%%%%%%%
\subsubsection{Hyperbolic Process}
%%%%%%%%%%%%%%%%%%%
As a final simulated example, we consider the two parameter hyperbolic process  driven by
$$
dS_t = -\frac{\kappa S_t}{\sqrt{1 + S^2_t}}dt + \sigma dW_t,
$$
 with  $\kappa, \sigma > 0$, which is a special case of the general hyperbolic diffusion of \cite{barndorff1978hyperbolic},
$$
dS_t = \frac{\sigma^2}{2}\left(\beta - \gamma \frac{ S_t}{\sqrt{\delta^2 + (S_t - \mu)^2}} \right)dt + \sigma dW_t.
$$
\begin{center}
%\FloatBarrier
\begin{table}[h!t!b!]
\centering
  \addtolength{\tabcolsep}{1pt}
\scalebox{.9}{
\begin{tabular}{llccccccc}
    %\multicolumn{9}{c}{European Call}\\\hline
    \hline
\toprule
 & &  \multicolumn{1}{c}{}&\multicolumn{2}{c}{CTMC-MLE} & \multicolumn{2}{c}{Kessler} & \multicolumn{2}{c}{Shoji-Ozaki}\\
   \cmidrule(r{1em}){4-5}\cmidrule(r{1em}){6-7}\cmidrule{8-9}
 $N$ & $1/\Delta$ &  True Param. & $\hat \theta_{N,m} - \theta $ & $\text{sd}(\hat \theta_{N,m})$& $\hat \theta_{N} - \theta $ & $\text{sd}(\hat \theta_{N})$ & $\hat \theta_{N} - \theta $ & $\text{sd}(\hat \theta_{N})$ \\
    \toprule
 120 & 24 & $\kappa=4.000$  &  -0.241 & 1.155  &  0.443 & 1.397   &  0.546 & 1.425  \\ 
  & & $\sigma=0.300$   &  -0.007 & 0.021  &  -0.023 & 0.018   &  0.006 & 0.022  \\ \hline
 260 & 52 & $\kappa=4.000$  &  0.218 & 1.288  &  0.327 & 1.379   &  0.342 & 1.365  \\ 
  & & $\sigma=0.300$   &  -0.001 & 0.014  &  -0.010 & 0.013   &  0.002 & 0.014  \\ \hline
  1250 & 250 & $\kappa=4.000$  &  0.422 & 1.335  &  0.430 & 1.329   &  0.430 & 1.320  \\ 
  & & $\sigma=0.300$   &  -0.000 & 0.006  &  -0.002 & 0.006   &  0.000 & 0.006  \\ 
    \bottomrule
\end{tabular}}
\caption{\small{Hyperbolic Process - comparison of CTMC-MLE vs Kessler and Shoji-Ozaki, with $m=300$ states. Results from 500 repeated simulations, with the same randomized initial guess. Initial $S_0=0.2$. Fixed time horizon $T=5$ with varying sampling frequency, $1/\Delta$.}} % title of Table
  \label{table:Hyperbolic}
\end{table}
\FloatBarrier
\end{center}
The estimates are summarized in Table \ref{table:Hyperbolic} for the three methods. We can see a clear advantage in terms of estimation error for the CTMC-MLE method (in terms of error and standard deviation), especially for the less frequent sampling. With daily sampling, all three methods perform similarly well.

%%%%%%%%%%%%%%%%%%%%%%%%%%%
\subsection{Real Data Example: Constant Maturity Interest Rates}
%%%%%%%%%%%%%%%%%%%%%%%%%%%
\begin{figure}[h!t!b!]
\centering     %%% not \center
\includegraphics[width=.95\textwidth]{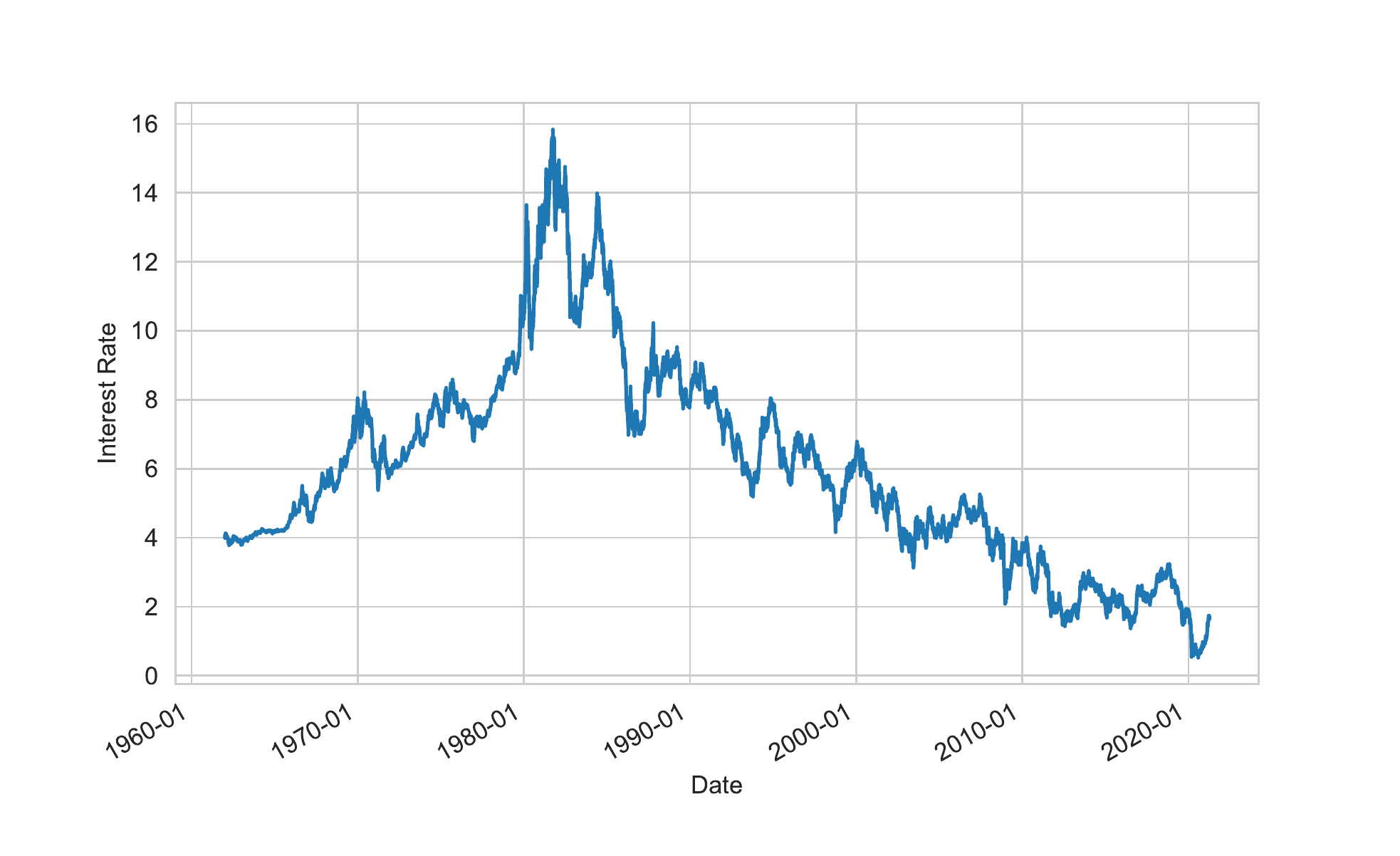}
\caption{\small{10-Year Treasury Constant Maturity Rates from 1962 to 2021.
}}\label{fig:OUPath}
\end{figure}
In this example, we fit the CKLS model of Section \ref{sect:CKLS} to a sample of historical interest rates over the period Jan 1, 1962 to April 8, 2021. The data consists of 14,801 daily observations of the 10-Year Constant Maturity rate.\footnote{Board of Governors of the Federal Reserve System (US), 10-Year Treasury Constant Maturity Rate [DGS10],
retrieved from FRED, Federal Reserve Bank of St. Louis; https://fred.stlouisfed.org/series/DGS10, April 11, 2021.}  Figure \ref{fig:OUPath} displays the historical daily time series. We fit the CKLS model family using the CMTC-MLE approach, along with three time-discretization benchmarks: Kessler, Shoji-Ozaki, and Euler.  Fits are obtained using the full daily sample, with a sampling frequency of $1/\Delta = 252$ business days per year ($N=14801$), as well as weekly ($N=2961$) and yearly sampling ($N=58$). The parameter estimates are displayed in Table \ref{table:CKLSReal} for each method.

\begin{center}
%\FloatBarrier
\begin{table}[h!t!b!]
\centering
  \addtolength{\tabcolsep}{1pt}
\scalebox{.95}{
\begin{tabular}{lcccccc}
    %\multicolumn{9}{c}{European Call}\\\hline
    \hline
\toprule
 $N$ & $1/\Delta$  &  Param. & CTMC & Kessler &Shoji-Ozaki & Euler \\
    \toprule
     58 & 1 & $\theta_1$  &0.161 & 0.002  &  0.226 & 0.147     \\ 
  &(yearly) & $\theta_2$   &  -0.023 & -0.016  &  -0.049 & -0.033     \\ 
  & & $\theta_3$   & 0.576 & 0.610  &  0.480 & 0.467     \\ 
  & & $\theta_4$   & 0.378 &  0.311  &  0.487 & 0.487     \\ \hline
%%%%%%%%%%%%%%%%%%%%%%%%%%%%%%%
 2961 & 52  & $\theta_1$  & 0.238 & 0.082  &  0.281 & 0.273     \\ 
  & (weekly)& $\theta_2$   &  -0.046 & -0.019  &  -0.053 & -0.052     \\ 
  & & $\theta_3$   & 0.491 & 0.497  &  0.498 & 0.498     \\ 
  & & $\theta_4$   & 0.431 & 0.426  &  0.427 & 0.427     \\ \hline
  %%%%%%%%%%%%%%%%%%%%%%%%%%%%%%%
   14801 & 252 & $\theta_1$  & 0.191 & 0.059  &  0.271 & 0.267     \\ 
  & (daily) & $\theta_2$   &  -0.038 & -0.016  &  -0.051 & -0.051     \\ 
  & & $\theta_3$   & 0.559 & 0.559  &  0.559 & 0.558     \\ 
  & & $\theta_4$   &  0.325 & 0.338  &  0.338 & 0.338     \\
    \bottomrule
\end{tabular}}
\caption{\small{CKLS model fit to 10-Year Constant Maturity interest rates.}} % title of Table
  \label{table:CKLSReal}
\end{table}
\FloatBarrier
\end{center}

While the ``true'' parameters are unknown (assuming that the rates data-generating process belongs to the CKLS parametric family), we can get an idea of the bias of the estimator based on how much its estimate changes as we reduce the sample size that it sees. Comparing the case of $N=14801$ data points to $N=58$, what stands out is how stable the CTMC estimate is, which reflects the fact that the CTMC approximation \emph{has no time discretization error}, while each of the other three methods are susceptible to this source of error. Moreover, we can see that the Euler and Shoji-Ozaki methods are very similar for all sample sizes, and it is well known that the Euler method suffers from a large time discretization bias.  The CTMC method is a viable alternative to these types of approximations which is especially attractive for situations in which time-discretization bias of the sample is concerning, such as with weekly, monthly, quarterly, or yearly sampled time series (all of which are typical for econometric data).

%%%%%%%%%%%%%%%%%%%%%%%%%%%
\subsection{Real Data Example: USD/Euro Exchange Rates}
%%%%%%%%%%%%%%%%%%%%%%%%%%%

\begin{figure}[h!t!b!]
\centering     %%% not \center
\includegraphics[width=.95\textwidth]{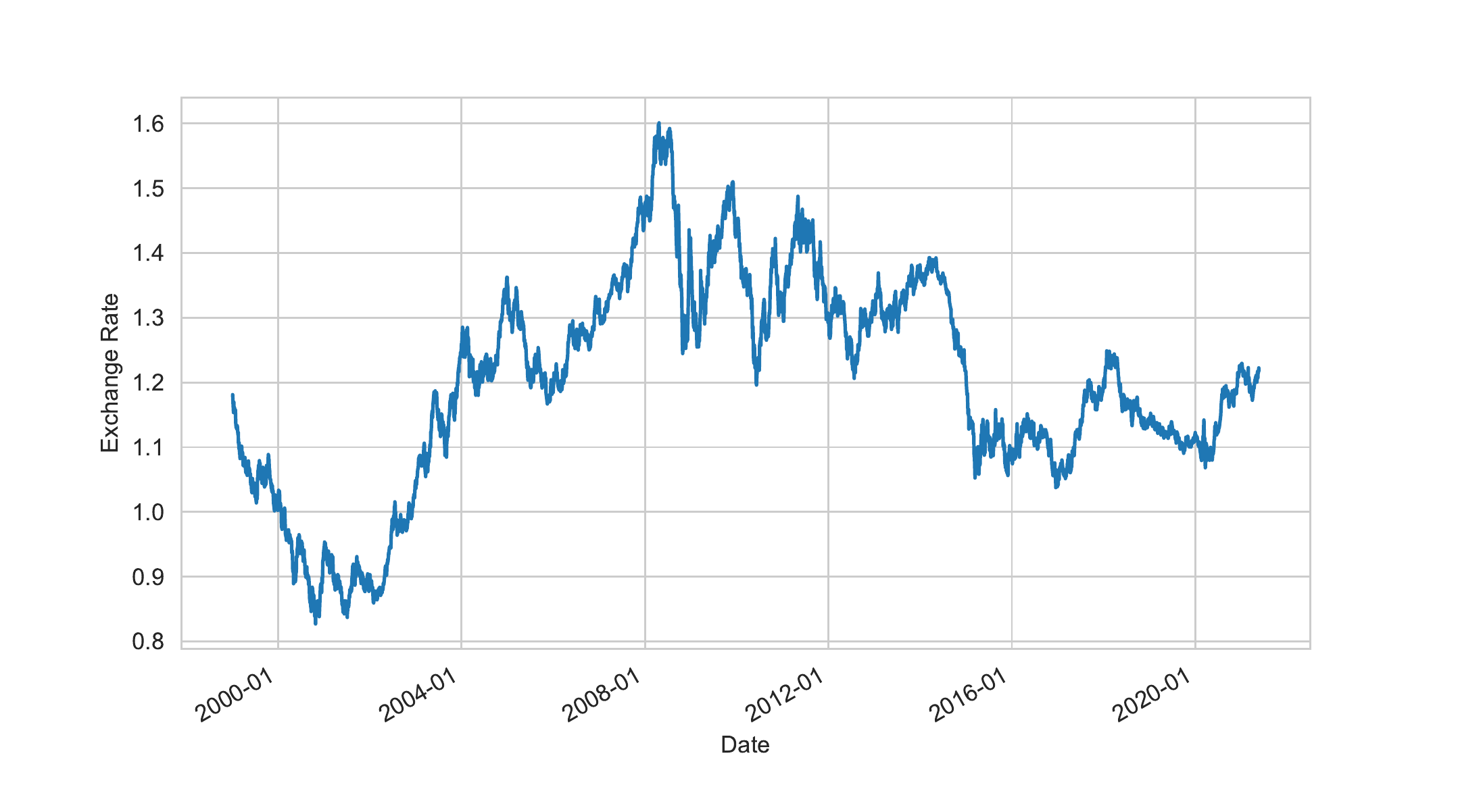}
\caption{\small{FX rates - USD/EUR from 1999 to 2021.
}}\label{fig:FXUSD}
\end{figure}

In this final example, we fit a time series of USD/EUR exchange rates over the period Jan 1, 1999 to May 21, 2021.\footnote{Board of Governors of the Federal Reserve System (US), U.S. / Euro Foreign Exchange Rate [DEXUSEU], retrieved from FRED, Federal Reserve Bank of St. Louis; https://fred.stlouisfed.org/series/DEXUSEU, May 24, 2021.}
The time series is displayed in Figure \ref{fig:FXUSD}, where it exhibits a clear mean-reverting pattern. The fitted parameters are provided in Table \ref{table:CKLSRealFX} for each method.  Taking the case of $1/\Delta = 252$ (which utilizes the full sample) as the consensus of parameter estimates, we can see that all methods are fairly consistent with each other, with the possible exception of Kessler's method.  Overall, these experiments demonstrate that the CTMC-MLE method is a reliable estimation approach for univariate diffusion models.

\begin{center}
%\FloatBarrier
\begin{table}[h!t!b!]
\centering
  \addtolength{\tabcolsep}{1pt}
\scalebox{.95}{
\begin{tabular}{lcccccc}
    %\multicolumn{9}{c}{European Call}\\\hline
    \hline
\toprule
 $N$ & $1/\Delta$  &  Param. & CTMC & Kessler &Shoji-Ozaki & Euler \\
    \toprule
   281   & 12 & $\theta_1$  &0.239 & 0.218 & 0.243 & 0.243 \\   
  &(monthly) & $\theta_2$   &  -0.200 & -0.179 & -0.200 & -0.200 \\ 
  & & $\theta_3$   & 0.097 & 0.097 & 0.099 & 0.098 \\ 
  & & $\theta_4$   & 0.933 & 0.869 & 0.930 & 0.903 \\ \hline
%%%%%%%%%%%%%%%%%%%%%%%%%%%%%%%
1124 & 52  & $\theta_1$  & 0.240 & 0.229 & 0.241 & 0.243 \\
  & (weekly)& $\theta_2$   & -0.200 & -0.187 & -0.199 & -0.200 \\
  & & $\theta_3$   & 0.099 & 0.096 & 0.099 & 0.099 \\
  & & $\theta_4$   & 0.784 & 0.782 & 0.796 & 0.796 \\\hline
  %%%%%%%%%%%%%%%%%%%%%%%%%%%%%%%
 5617  & 252 & $\theta_1$  & 0.240 & 0.216 & 0.243 & 0.243 \\
  & (daily) & $\theta_2$   & -0.200 & -0.177 & -0.200 & -0.200 \\
  & & $\theta_3$   & 0.095 & 0.095 & 0.095 & 0.095 \\
  & & $\theta_4$   &  0.960 & 0.982 & 0.986 & 0.986 \\
    \bottomrule
\end{tabular}}
\caption{\small{CKLS model fit to USD/EUR exchange rates from 1999 to 2021.}} % title of Table
  \label{table:CKLSRealFX}
\end{table}
\FloatBarrier
\end{center}

%%%%%%%%%%%%%%%%%%%%%%%%%%%%
%\subsection{Real Data Example: CBOE VIX}
%%%%%%%%%%%%%%%%%%%%%%%%%%%%
%
%\begin{figure}[h!t!b!]
%\centering     %%% not \center
%\includegraphics[width=.95\textwidth]{./images/VIX.eps}
%\caption{\small{CBOE VIX from 1990 to 2021.
%}}\label{fig:VIX}
%\end{figure}

%\section{Adding Jump into the model}
%\red{Let see if we can do this}
%\section{Two dimension diffusion}
%We will apply the ideas in the SIAM paper in this case.
%\subsection{Time-inhomogeneous Example}
%\noindent\textit{(Extended Black-Scholes model)}: In this model,
%the dynamics of $S_t$ is given by
%$$
%dS_t=\mu S_tdt+\sigma_0e^{\sigma_1 t}S_tdW_t,
%$$
%It can be shown (\cite{choi2015explicit}) that the transition density function $p(s^\prime, t^\prime|s,t)$ 
%is log-normal with mean $S_0e^{\mu t}$ and variance 
%$
%S_0^2e^{2\mu t}\left( \frac{\sigma_0^2}{2\sigma_1}(e^{2\sigma_1 t}-1)-1\right)
%$
%\red{We can handle this model now}

\section{Conclusion}
\label{section:Conclusion}
We propose a novel continuous-time Markov chain
approach to estimate
the unknown parameters
of general one-dimensional diffusions.
By utilizing a spatial discretization approach, the method avoids time-discretization error, and is thus safely applicable for time-series with all sampling frequencies.  The CTMC structure enables us to obtain likelihood approximations in closed-form, thus facilitating maximum likelihood estimation.  Comparisons with existing estimators (Exact MLE, Euler, Kessler, and Shoji-Ozaki) demonstrate the favorable performance of this new parameter estimation framework.
It will be interesting to extend the approach
proposed in this paper to higher dimensional diffusions. We leave
this as an interesting research problem for future studies.
\bibliographystyle{amsalpha}
\bibliography{LV}

\appendix

\end{document}